%% file: main.tex
\pdfoutput=1

\documentclass[11pt]{article}

\usepackage[numbers]{natbib}
\usepackage[section]{placeins}
\usepackage{xspace}
\usepackage{xcolor}
\definecolor{darkblue}{rgb}{0,0,.5}
\usepackage[colorlinks=true,allcolors=darkblue]{hyperref}
\usepackage{algorithm}
\usepackage[noend]{algpseudocode}
\usepackage{bbm}
\usepackage{macros}
\usepackage{jinshuomacros}
\usepackage{tikz}
\usetikzlibrary{matrix}
\usepackage{pgfplots}
\usepackage{overpic}
\usepackage{makecell}
\usepgfplotslibrary{fillbetween}
\usetikzlibrary{patterns}
\usepackage{authblk}
\usepackage{enumerate}
\usepackage{amssymb}
\usepackage{amsbsy}
\usepackage{amsmath}
\usepackage{amsthm}
\usepackage{amsfonts}
\usepackage{latexsym}
\usepackage{graphicx}
\usepackage{color}
\usepackage{xifthen}
\usepackage{algorithm}
\usepackage{algorithmicx, algpseudocode}
\usepackage{thmtools}
\usepackage{thm-restate}
\usepackage{esvect}

\usepackage{latexsym}

\usepackage{subcaption}

\makeatletter
\long\def\@makecaption#1#2{
  \vskip 0.8ex
  \setbox\@tempboxa\hbox{\small {\bf #1:} #2}
  \parindent 1.5em  
  \dimen0=\hsize
  \advance\dimen0 by -3em
  \ifdim \wd\@tempboxa >\dimen0
  \hbox to \hsize{
    \parindent 0em
    \hfil 
    \parbox{\dimen0}{\def\baselinestretch{0.96}\small
      {\bf #1.} #2
    } 
    \hfil}
  \else \hbox to \hsize{\hfil \box\@tempboxa \hfil}
  \fi
}
\makeatother

\newcommand{\R}{\mathbb{R}}

\newtheorem{theorem}{Theorem}
\newtheorem{lemma}{Lemma}[section]
\newtheorem{corollary}{Corollary}[section]

\newtheorem{proposition}{Proposition}
\newtheorem{definition}{Definition}[section]
\newtheorem{claim}{Claim}[section]
\newtheorem{fact}{Fact}

\newtheorem{example}{Example}

\renewcommand\log{\ln}

\usepackage{fullpage}

\begin{document}

\title{Optimal Differential Privacy Composition for Exponential Mechanisms and the Cost of Adaptivity}
\author[1]{Jinshuo Dong\thanks{Work done while at LinkedIn}}
\author[2]{David Durfee}
\author[2]{Ryan Rogers}
\affil[1]{Applied Mathematics and Computational Sciences, University of Pennsylvania}
\affil[2]{Applied Research, LinkedIn}

\maketitle 

\begin{abstract}
\input{abstract}

\end{abstract}

\clearpage

\tableofcontents

\clearpage

\input{intro}

\input{prelims}

\input{techniques}

\input{reduction}

\input{non-interactive}

\input{interactive-gap}

\input{mgf}

\input{conclusion}

\input{acks}

\clearpage

\bibliography{bib}
\bibliographystyle{abbrvnat}

\clearpage

\appendix

\input{gen-rand-response}

\input{non-interactive-proofs}

\input{gap-proofs}

\end{document}

%% file: abstract.tex
Composition is one of the most important properties of differential privacy (DP), as it allows algorithm designers to build complex private algorithms from DP primitives.  We consider precise composition bounds of the overall privacy loss for exponential mechanisms, one of the most fundamental class of mechanisms in DP.  We give explicit formulations of the optimal privacy loss for both the adaptive and nonadaptive settings.  For the nonadaptive setting in which each mechanism has the same privacy parameter, we give an efficiently computable formulation of the optimal privacy loss.  Furthermore, we show that there is a difference in the privacy loss when the exponential mechanism is chosen adaptively versus nonadaptively. To our knowledge, it was previously unknown whether such a gap existed for any DP mechanisms with fixed privacy parameters, and we demonstrate the gap for a widely used class of mechanism in a natural setting. We then improve upon the best previously known upper bounds for adaptive composition of exponential mechanism with efficiently computable formulations and show the improvement.

%% file: intro.tex

\section{Introduction}

Differential privacy (DP) has emerged as the leading privacy benchmark in machine learning as well as data analytics on sensitive data sets.  One of the most useful properties of DP is that it composes, with slight degradation in the overall privacy loss parameters.  This allows algorithm designers to build complicated algorithms whose privacy analysis follows from the fact that each subroutine satisfies DP.  Further, composition allows us to bound the amount of privacy loss, quantified by the $(\diffp,\delta)$ parameters in DP, consumed by an (adaptive) sequence of DP algorithms evaluated on the same dataset.  Hence, there have been several works in DP that help bound the privacy loss in composition, starting with basic composition from \citet{DworkMcNiSm06} and advanced composition from \citet{DworkRoVa10}.  More recently, there have been works that give the exact, optimal privacy loss bound when all that is known is that the individual algorithms are each DP:  \citet{KairouzOhVi17} give the optimal privacy loss bound in the homogeneous case, where all the privacy parameters for each algorithm are the same, and \citet{MurtaghVa16} give the more general optimal privacy loss bound in the heterogeneous case, where all the privacy parameters can be different at each round.

Although these \emph{black box} composition theorems give the best possible bound on the privacy loss over multiple rounds of general DP algorithms, one should be able to improve on this bound when considering specific subclasses of DP algorithms.  One example of such a composition theorem that takes into account the particular algorithm being used at each round is in \emph{moments accounting composition} from \citet{Abadietal16}.  For their setting, they use noisy stochastic gradient descent and account for the subsampling and Gaussian noise that is added to the gradients at each round in their overall privacy loss bound.  In particular, they are able to save a factor of $O(\sqrt{\log(k/\delta)})$ in the overall privacy parameter, where $k$ is the number of gradient descent steps taken.  
Another example of \emph{white box} composition is from \citet{DurfeeRo19} who introduce \emph{bounded range} (BR) as a property for DP algorithms, which leads to improved composition bounds compared to the general case optimal bound.  

Arguably, the fundamental DP primitives are the following: randomized response \cite{Warner65}, Laplace mechanism \cite{DworkMcNiSm06} or its discretized variant (geometric mechanism), Gaussian mechanism \cite{DworkKeMcMiNa06}, and the exponential mechanism \cite{McSherryTa07}. The optimal DP composition bounds \cite{KairouzOhVi17,MurtaghVa16} follow by showing that each DP algorithm, once the neighboring datasets are fixed, can be written as randomized response composed with a post-processing function that is independent of the data.  Hence, the optimal DP composition is essentially tailored to composing randomized response mechanisms.  The geometric mechanism was shown to also achieve the optimal composition bound \cite{KairouzOhVi17}. Optimal DP composition bound for Gaussian mechanisms is obtained as a special case of the general composition theorem in \citet{DongRoSu19}.

Hence, it is only natural to then ask:
{\bf what is the optimal DP composition bound over the class of exponential mechanisms?}  
This question is the primary focus of this work.  As was shown in \citet{DurfeeRo19}, the exponential mechanism satisfies the BR property and hence enjoys their improved composition bound.   The exponential mechanism provides a very general way to design DP algorithms over an arbitrary outcome space where a \emph{quality score} measures the value of each possible outcome given the input datatset. In practice, the exponential mechanism is most often deployed when a maximum or minimum operation is needed in a DP algorithm.  

Surprisingly, the answer to this question depends on whether the choice of exponential mechanism is adaptively chosen at each round or not.  For the existing DP composition bounds, adaptivity in the choice of DP algorithm did not affect the overall privacy bound, even in the optimal privacy loss bounds.  \citet{RogersRoUlVa16} show that there is an asymptotic gap in the privacy loss bound when the privacy parameters $\{ \diffp_i\}_{i = 1}^k$ are fixed in advance versus when an analyst can adaptively select the privacy parameters $\diffp_i$ at each round $i$ based on previous outcomes before $i$. However, we focus on the traditional view of DP that fixes all the privacy parameters up front.  

In the local setting of differential privacy, interactivity and adaptivity have been shown to be significant in learning algorithms and estimation tasks, see \cite{KasiviswanathanLeNiRaSm11, SmithThUp17, JosephMaNeRo19}, although for some estimation tasks in more restricted interactivity models, there is no distinction \cite{DuchiRo19}.  However, such interactivity models are not relevant to the central model since mechanisms are designed to take the full dataset as input rather than designing algorithms on each datum as in the local model. Our result is in a similar vein to these  results that consider the possible impact to the privacy loss from giving the adversary additional power.

We find the gap here particularly interesting because this is such a natural setting and has practical interest in the deployment of top-$k$ algorithms \cite{DurfeeRo19}. 
For such data queries, it would be ideal to allow the analyst to adaptively interact with a DP system, rather than having the analyst select all the mechanisms up front and produce results as a batch.  
For example, consider the exponential mechanism as simply reporting the (noisy) maximum index for some metric of interest, but only for a certain subgroup and the analyst specifies the classifications for this subgroup, such as company, job title, geographic location, etc.
Even if we fix the privacy parameter, our privacy loss will increase if we allow the analyst to adaptively select these classifications in subsequent queries.

Both the nonadaptive and adaptive setting will have practical importance and the distinction will be important in efficiently computing the tightest possible bounds on the privacy loss.
In particular, our nonadaptive and efficiently computable composition formulation can be applied in a \emph{dashboard setting}, where the set of queries that are privately output for a dataset is predetermined, and could include top-$k$ queries for all the metrics of interest.
Further, we know that our bound cannot be improved in this setting.  Alternatively, our efficiently computable improved bounds for the adaptive setting
can be applied in an \emph{API setting} mentioned above, where the analyst adaptively interacts with the DP system.

While the improvements we give here in bounding the overall privacy loss are not asymptotically significant, if we consider the amount of privacy loss to be fixed, then increasing the number of allowable queries by a constant factor can still have a substantial impact on practical deployments. 
From our results in Figure~\ref{fig:ALL}, our nonadaptive composition bound allows for about four times more queries than the optimal composition for general DP mechanisms given a fixed privacy loss budget. Furthermore, this optimal composition allows for about two times more queries than the improved bounds given in \cite{DurfeeRo19}. Additionally, in some settings our improvement for the adaptive composition bound of exponential mechanisms allows for about three times more queries than both the optimal composition for DP mechanisms and the improved bounds in \cite{DurfeeRo19}.

\subsection{Our contributions}

We informally summarize our main contributions here and will give the formal statements in Section~\ref{sec:results} once we have set up the requisite notation.

\begin{enumerate}
\item We show that there is indeed a gap between the optimal composition bound when an adversary can adaptively select different exponential mechanisms at each round as opposed to an adversary who selects all the exponential mechanisms in advance.  This is in contrast to traditional DP composition bounds, which showed no difference between these different adversaries in terms of the privacy loss.
\item For the nonadaptive adversary, we provide an explicitly computable formula for the optimal composition bound that can be computed in $O(k^2)$ time, where $k$ is the total number of exponential mechanisms that are executed.  
\item For the adaptive adversary, we provide an explicit formulation for the optimal composition, but in a recursive formulation that is intractable 
to compute for even reasonably sized $k$.
We then improve upon the previous upper bound on the privacy loss by giving an improved KL divergence bound, and further provide a numerical scheme based on the moment generating function of the privacy loss to obtain an even better upper bound on the optimal composition.  \end{enumerate}

Although we have presented the exponential mechanism as a specific DP mechanism, it is also important to discuss its generality.  
In particular, there is the folklore result that states that any (pure) DP mechanism can be written in terms of an exponential mechanism with a particular quality score, i.e. the log-density of the mechanism \citep{McSherryTa07}.  
Hence, it might seem that the optimal $k$-fold adaptive composition bound over the class of exponential mechanisms, or BR mechanisms, is also the optimal $k$-fold adaptive composition bound over the class of all DP mechanisms. 
However, sometimes taking general DP mechanisms, such as randomized response or the Laplace mechanism, to the generic exponential mechanism form could result in a different overall privacy parameter.  Hence, a general $\diffp$-DP mechanism can be written in terms of an exponential mechanism with parameter $\diffp'$, which can be up to a factor 2 larger than $\diffp$.
See Section~\ref{sec:OptCompVsBR} for more discussion.

%% file: prelims.tex
\section{Preliminaries} 

In this section, we set up the necessary notation and definitions for our results. It will be necessary in our analysis to use a generalized version of randomized response that corresponds to BR mechanisms. Similar to the work in the optimal composition bounds for DP mechanisms, our goal will be to reduce composition to adaptive calls of this more generalized randomized response than the one used in the optimal DP composition analysis \cite{KairouzOhVi17,MurtaghVa16}. For this reduction, we give a more fine-grained definition of adaptive composition, that is equivalent to previous versions, but includes details that were not necessary for standard DP composition.  In particular, the class of algorithms that we want to give a DP composition bound for is not closed under convex combinations.  Thus, an adversary can randomize over different algorithms in the same class and the resulting algorithm is no longer in that class.  
Finally, we give the definition of optimal composition and an alternative formulation that will be easier to work with.

We first cover the standard differential privacy definition from \cite{DworkMcNiSm06,DworkKeMcMiNa06}, where we will say that two datasets $x,x' \in \cX$ are neighbors if they differ in the addition or deletion of one individual's data, sometimes denoted as $x \sim x'$.
\begin{definition}\label{def:dp}
	A mechanism $M:\cX\to \cY$ is $(\diffp, \delta)$-differentially-private (DP) if the following holds for any neighboring dataset $x,x'$ and $S \subseteq \cY$:
	\[
		{\Pr[M(x) \in S]}\leqslant
		\e^{\diffp}{\Pr[M(x') \in S]} + \delta.
	\]
Also if $\delta= 0$, we simply write $\diffp$-DP.
\end{definition}

We now present the definition of bounded range from \citet{DurfeeRo19}, which was useful in improving the composition bounds for their algorithms.

\begin{definition}
	A mechanism $M:\cX\to \cY$ is $\diffp$-bounded-range (BR) if the following holds for any neighboring dataset $x,x'$ and $y_1,y_2\in \cY$:
	\[
		\frac{\Pr[M(x)=y_1]}{\Pr[M(x')=y_1]}\leqslant
		\e^\diffp\frac{\Pr[M(x)=y_2]}{\Pr[M(x')=y_2]}.
	\]
\end{definition}
Note that for continuous outcome spaces, we can use the probability density function instead.  We then have the following equivalent formulation of BR mechanisms, which will be easier to use in our analysis.
\begin{corollary}\label{cor:br}
A mechanism $M: \cX \to \cY$ is $\diffp$-BR if and only if for any neighboring databases $x,x'$ there exists some $t \in [0,\diffp]$ such that for any outcome $y \in \cY$ we have

\[
t - \diffp \leq \log\left(\frac{\Pr[M(x) = y]}{\Pr[M(x') = y]}\right) \leq t
\]
\end{corollary}
We also have the following connection between BR and (pure) DP.
\begin{lemma}[Corollary 4.2 in \cite{DurfeeRo19}]\label{lem:factor2BR}
If $M$ is $\diffp$-BR then it is $\diffp$-DP.  Furthermore, if $M$ is $\diffp$-DP, then it is also $2\diffp$-BR.
\end{lemma}

We will now define the exponential mechanism in terms of a quality score $u: \cX \times \cY \to \R$ and its sensitivity $\Delta u \defeq \max_{y \in \cY} \max_{x \sim x'} | u(x,y) - u(x',y)|$.
\begin{definition}[Exponential Mechanism \cite{McSherryTa07}\label{defn:em}]
The exponential mechanism $M_u: \cX \to \cY$ samples an outcome $y \in \cY$ with probability proportional to $\exp\left( \tfrac{\diffp u(x,y)}{2\Delta u} \right)$.
\label{defn:exp_mech}
\end{definition}
The factor of two accounts for the possibility that the normalization term can also change with a neighboring dataset and for some quality scores, i.e. monotonic, the factor of 2 is not necessary.  We then have the following result.
\begin{theorem}
The exponential mechanism is $\diffp$-DP \citep{McSherryTa07}.  Further, the exponential mechanism is $\diffp$-BR \cite{DurfeeRo19}.
\end{theorem}

Throughout the rest of this work, we will use a generalized version of randomized response, which our analysis will primarily focus on and we define below.
\begin{definition}[Generalized Random Response]\label{defn:gen_rr}
For any $\diffp \geq 0$ and $t \in [0,\diffp]$,
let $\grr{\diffp,t}: \{0,1\} \rightarrow \{0,1\}$ be a randomized mechanism in terms of probabilities $q_{\diffp,t}$ and $p_{\diffp,t}$ such that

\begin{align*}
\grr{\diffp,t}(0) = 0 \text{ w.p. } \frac{1 - e^{t-\diffp}}{1 - e^{-\diffp}} \eqdef q_{\diffp,t} \qquad & \text{ and } \qquad \grr{\diffp,t}(0) = 1 \text{ w.p. } \frac{e^{t-\diffp} - e^{-\diffp}}{1 - e^{-\diffp}} \eqdef 1 - q_{\diffp,t}
\\
\grr{\diffp,t}(1) = 0 \text{ w.p. } \frac{e^{-t} - e^{-\diffp}}{1 - e^{-\diffp}} \eqdef p_{\diffp,t} \qquad & \text{ and } \qquad \grr{\diffp,t}(1) = 1 \text{ w.p. } \frac{1 - e^{-t}}{1 - e^{-\diffp}} \eqdef 1 - p_{\diffp,t}.
\end{align*}
\end{definition}

Note the $\grr{\diffp, \diffp/2}$ is simply the standard randomized response with privacy parameter $\diffp/2$ \cite{Warner65}.  
We will typically drop the dependence of $\diffp$ in $\grr{\diffp,t} \equiv \grr{t}$ when it is clear from context.
It will be useful to also define what we mean by optimal privacy parameters, which we will write as a function $\delta_{\opt}$ of a mechanism and a global DP parameters $\diffp_g$.  

\begin{definition}[Optimal Privacy Parameters]\label{def:optimal_delta}
Given a mechanism $M: \cX \to \cY$ and any $\diffp \in \R$, we define the optimal $\delta$ to be

\[
\delta_{\opt}(M,\diffp) \defeq \inf \big\{ \delta: \text{ M is } (\diffp,\delta)\text{-DP}\big\}
\]

Further, if $\cM$ is a class of mechanisms $M: \cX \to \cY$, then for any $\diffp \in \R$, we define 

\[
\delta_{\opt}(\cM,\diffp) \defeq \sup_{M \in \cM} \delta_{\opt}(M,\diffp)
\]

\end{definition}

\begin{fact}\label{fact:delta_opt}
For any mechanism $M: \cX \to \cY$ and $\diffp \in \R$
\begin{equation}
\delta_{\opt}(M,\diffp) =  \sup_{x\sim x' } \int_{y \in \cY} \max\{\Pr[M(x) = y] - e^{\diffp}\Pr[M(x') = y], 0\} dy
\label{eq:opt_delta}
\end{equation}

\end{fact}

\begin{proof}
It follows immediately from definition that $M$ is $(\diffp,\delta)$-DP if and only if

\[
\sup_{x\sim x' } \sup_{S \subseteq \cY} \left\{ {\Pr[M(x) \in S]} - 
\e^{\diffp}{\Pr[M(x') \in S]} \right\} \leq \delta
\]

This immediately implies

\[
\sup_{x\sim x' } \sup_{S \subseteq \cY} \left\{ {\Pr[M(x) \in S]} - 
\e^{\diffp}{\Pr[M(x') \in S]} \right\} = \delta_{\opt}(M,\diffp) 
\]

Furthermore, it is straightforward to see that for any neighbors $x,x'$

\[
\sup_{S \subseteq \cY} \left\{ {\Pr[M(x) \in S]} - 
\e^{\diffp}{\Pr[M(x') \in S]} \right\} = \int_{y \in \cY} \max\{\Pr[M(x) = y] - e^{\diffp}\Pr[M(x') = y], 0\} dy
\]

\end{proof}

\newcommand{\esssup}{\mathrm{ess}\sup}
\newcommand{\essinf}{\mathrm{ess}\inf}
\newcommand{\sens}{\mathrm{sens}}
\newcommand{\range}{\mathrm{range}}

\subsection{Improved semantics for the exponential mechanism\label{sub:semantics_of_exponential_mechanism}} 

Here we present a slight modification to the traditional exponential mechanism presented in Definition~\ref{defn:exp_mech}.  In particular, rather than presenting the probability of selecting different outcomes in terms of the quality score's \emph{sensitivity}, we define it in terms of what we call the \emph{range} of the quality score.  This leads to a simpler formulation of the exponential mechanism that does not have to consider whether a quality score is monotonic or not, i.e. whether to include a factor of two or not in the probability sampling rate, and for this reason we only view our modification as a semantic improvement.
Additionally, we present the following example, to show that defining the exponential mechanism in terms of the max sensitivity leads to unwanted properties, which suggests that sensitivity is not a canonical parameter that should appear in the exponential mechanism.

\begin{example}
Let $u: \cX \times [m] \to \R$ be an arbitrary quality score with sensitivity $\Delta u$. Consider an arbitrary function $f:\cX\to \R$ on the data domain. We define the alternate quality score $u'(x,i) \defeq u(x,i)+f(x)$. It is easy to see that
\[
\frac{e^{\diffp u(x,i)}}{\sum_i e^{\diffp u(x,i)}} = \frac{e^{\diffp u'(x,i)}}{\sum_i e^{\diffp u'(x,i)}}.
\]
That is, the privacy properties of the exponential mechanism with quality score $u$ and $u'$ are equivalent. 
However, it is very common that $\Delta u \neq \Delta u'$. For example let $X = Y = \{0,1\}$ and $u(x,y) = x+y$, $f(x)=10x$ and hence $u'(x,y)=11x+y$. Clearly, $\Delta u=1$ and $\Delta u' =11$.
\end{example}

Note that this example is carefully constructed to show that using sensitivity has unwanted properties and we found no examples of such utility functions used in the literature. However, it would be nice to have a definition that also optimally handles such utility functions, in addition to encapsulating the monotonic case in the definition. 

Given a quality score $u:\cX \times \cY \to \R$, consider defining the exponential mechanism in terms of some function of the quality score $\phi(u)$, e.g. $\phi(u) \equiv 2\Delta u$ would give us the traditional exponential mechanism.  Instead, let's consider the property that $\phi(u)$ needs to satisfy to ensure a mechanism $M: \cX \to \cY$ is $\diffp$-BR, and hence $\diffp$-DP. Let $x,x' \in \cX$ be neighbors and fix outcomes $y,y' \in \cY$. To ensure $\diffp$-BR, we require the following condition on $\phi(u)$ (note that the normalization factors cancel)
\[
 \frac{ \exp\left(\tfrac{\diffp u(x,y)}{\phi(u)} \right)}{\exp\left(\tfrac{\diffp u(x',y)}{\phi(u)} \right)} \leq  \e^\diffp\cdot \frac{\exp\left(\tfrac{\diffp u(x,y')}{\phi(u)} \right)}{\exp\left(\tfrac{\diffp u(x',y')}{\phi(u)} \right)} \qquad
\Leftrightarrow \qquad  \text{ } u(x,y) - u(x',y) - u(x,y') + u(x',y') \leq \phi(u).
\]

With this observation, we define the \emph{range} $\widetilde\Delta u$ of a function $u$ as the following

\begin{align*}
\widetilde\Delta u &\defeq \sup_{x\sim x',y,y'\in \cY} u(x,y) - u(x',y) - u(x,y') + u(x',y')\\
&= \sup_{x\sim x'} \left\{ \max_y \left\{ u(x,y)-u(x',y) \right\} - \min_{y'}\left\{u(x,y')-u(x',y') \right\} \right\}
\end{align*}
We then have the following useful properties of the range.
\begin{proposition} The range $\widetilde\Delta u$ of a function $u: \cX \times \cY \to \R$ has the following properties
	\begin{itemize}
		\item $\widetilde\Delta u= \widetilde\Delta u'$ when $u'(x,y) = u(x,y)+f(x)$
		\item $\widetilde\Delta u \leqslant 2\cdot \Delta u$.
		\item $\widetilde\Delta u = \Delta u$ if $u$ is monotone.
	\end{itemize}
\end{proposition}

We then have the immediate result, which presents a variant of the exponential mechanism in terms of the range, rather than the sensitivity, of the quality score.
\begin{proposition}\label{prop:semantics}
	The mechanism $M: \cX \to \cY$ that samples $y \in \cY$ with probability proportional to $\exp\left(\dfrac{\diffp u(x,y)}{\widetilde\Delta u}\right)$ is $\diffp$-BR, and hence $\diffp$-DP.
\end{proposition}

\subsection{Formally defining composition}

We now present the definition of adaptive composition for DP algorithms in the setting introduced by \citet{DworkRoVa10}.
Our definition will be slightly more explicit in how we formulate the adversary.
Specifically, the adaptive composition game in \cite{DworkRoVa10} does not explicitly allow the analyst to use its own personal randomness in picking the mechanism at each round. Defining the analyst in this way is fine if the analyst selects a DP mechanism at each round, since we know that any convex combination of $\diffp$-DP mechanisms is still $\diffp$-DP. For example, if $M'$ and $M''$ are arbitrary $\diffp$-DP mechanisms, then if we define the mechanism $M$ to run $M'$ with probability $p$ and run $M''$ with probability $(1-p)$, the mechanism $M$ is $\diffp$-DP.  Therefore, any randomness used by the adversary in their choice of $\diffp$-DP mechanism can simply be simulated by another $\diffp$-DP mechanism and can be ignored in the definition.  

In full generality, the class of mechanisms that we allow the analyst to select from at each round may not necessarily be closed under convex combinations.  In particular, we will be considering the setting in which the class of mechanisms the adversary can choose from is restricted to $\diffp$-BR mechanisms, which is not closed under convex combinations, see Section~\ref{sec:convex_combo}.  Hence, in the adaptive composition game $\AdaComp$ presented in Algorithm~\ref{algo:adaptgame}, we decompose the adversary into a randomized and deterministic component.  The adversary will then use personal randomness $\cR$ at each round and based on this will then use a deterministic function $\cDe$ to select a mechanism $M_i$ from the class of algorithms $\cE_i$ at round $i$.\footnote{Similarly, \citet{RogersRoUlVa16} defined a \emph{simulated} game which explicitly decomposed the adversary into a deterministic post-processing function of randomized response at each round and then used randomness at the beginning of all the interactions to simulate the adaptive randomness at each round. They showed that such a simulated game is equivalent to the adaptive parameter composition game, which allowed them to simply consider randomized response mechanisms at each round with a deterministic adversary.}  As one would expect and we will show, the adversary cannot add their own independent randomness that is data-independent and further degrade privacy. 
The output of the adaptive composition game will be a sequence of random coins the adversary uses and the outcomes from applying the mechanism for the corresponding databases (given the bit $b$), which we write as $R_0, A^b_1, \cdots , R_{k-1}, A^b_k, R_k$.

\begin{algorithm}
\caption{$\AdaComp(\cA=(\cR,\cDe),(\cE_1,\cdots, \cE_k),b)$, where $\cDe$ is a deterministic algorithm, $\cR$ is a randomized algorithm, $\cE_1,\cdots, \cE_k$ are classes of randomized algorithms, and $b \in \{0,1\}$. \label{algo:adaptgame}}
\begin{algorithmic}
\State $r_0 \sim \cR(\emptyset)$.
\For{$i = 1,\cdots, k$}
\State {$\cDe(r_0,A^b_1,\cdots, r_{i-1})$  selects neighboring datasets $x^{i,0}, x^{i,1}$, and $M_i \in \cE_i$}
\State {$\cA$ receives $A_i^b = M_i(x^{i,b})$ }
\State {$r_i \sim \cR(r_0,A^b_1,\cdots, r_{i-1},A^b_{i})$}
 \EndFor
 \Return{ view $V^b = (r_0, A^b_1, r_1, \cdots , r_{k-1}, A^b_k, r_k)$}
\end{algorithmic}
\end{algorithm}

\begin{definition}[$k$-fold Adaptive Composition]
Given classes of randomized algorithms $\vcE = (\cE_1,\cdots \cE_k)$, we say $\vcE$ is $(\diffp_g,\delta_g)$-DP under $k$-fold adaptive composition if for any adversary $\cA$ and $b \in \{0,1\}$, along with any set $S$ that is a subset of outputs of $\AdaComp(\cA,\vcE,\cdot)$, we have

\[
\Pr[\AdaComp(\cA,\vcE,b) \in S] \leq e^{\diffp_g}\Pr[\AdaComp(\cA,\vcE,1-b) \in S] + \delta_g.
\]

\end{definition}

It will be important to distinguish adaptive and nonadaptive adversaries in our composition bounds.  The nonadaptive adversary selects all the mechanisms and neighboring datasets to be used at each round prior to any computation on the dataset.  For this case, we can simply study the privacy guarantees of a mechanism $M = M_1 \times M_2 \times \cdots \times M_k$ where each $M_i$ is $\diffp_i$-BR and $M(x) = (M_1(x), M_2(x), \cdots, M_k(x))$.

%% file: techniques.tex

\section{Overview of results and techniques} \label{sec:results}

Given the necessary notation and setup, we present formal statements of our main results along with the intuition and techniques used to prove these results.  We detail the formal proofs in the sequel.  
\subsection{Reduction to generalized randomized response}

Similar to \cite{KairouzOhVi17,MurtaghVa16}, we first want to identity the ``worst-case" mechanism for the class of BR mechanisms, which is to say that any BR mechanism can be simulated through post-processing of this worst-case mechanism. It then follows that composition over the class of BR mechanisms can be reduced to simply considering composition of this worst-case mechanism, allowing for explicit computation of the optimal composition. For the class of $\diffp$-DP mechanisms, the worst-case mechanism was shown to be randomized response through both the hypothesis testing interpretation \cite{KairouzOhVi17}, and explicitly constructing the post-processing function \cite{MurtaghVa16}. Rather than consider the class of exponential mechanisms in our analysis, we will instead focus on the more general class of BR mechanisms due to the fact that the BR property in Corollary~\ref{cor:br} closely matches the definition of (pure) DP.  
We also show in Section~\ref{sec:exp_to_grr} that this definition is essentially equivalent to the standard use of the exponential mechanism, which is to say that the privacy loss is identical for the worst-case BR mechanism and the exponential mechanism.
We then categorize the worst-case BR mechanisms similarly to analysis done in \cite{KairouzOhVi17,MurtaghVa16}. 
More specifically, we know from Corollary~\ref{cor:br} that if a mechanism $M: \cX \rightarrow \cY$ is $\ep$-BR, then for any neighboring $x, x'$ there exists some $t \in [0,\ep]$ such that for any $y \in \cY$,

\[
t - \diffp \leq \log\left(\frac{\Pr[M(x) = y]}{\Pr[M(x') = y]}\right) \leq t.
\]

Note that if for each neighboring $x,x'$ we have that $t = \ep/2$, then $M$ is also $\frac{\ep}{2}$-DP. It then follows from \cite{KairouzOhVi17,MurtaghVa16} that when $t = \ep/2$ the worst-case mechanism is simply randomized response with parameter $\frac{\ep}{2}$. Intuitively, this is the mechanism $M$ such that for any $y \in \cY$ one of the bounds is tight, in other words

\[
 \log\left(\frac{\Pr[M(x) = y]}{\Pr[M(x') = y]}\right) \in \left\{ -\frac{\ep}{2} , \frac{\ep}{2}\right\}.
\]

For our setting, this same intuition must hold for $t = \ep/2$, and we then generalize this to any $t \in [0,\ep]$ where the worst-case mechanism $M$ is such that 

\[
 \log\left(\frac{\Pr[M(x) = y]}{\Pr[M(x') = y]}\right) \in \left\{ t - \ep, t \right\}.
\]

This generalization is exactly our Definition~\ref{defn:gen_rr}, and using a similar interpretation to hypothesis testing, we show that for any given $t \in [0,\ep]$ this is the worst-case mechanism that satisfies the BR property. While this result is largely unsurprising, in Section~\ref{sec:reduction} we give a thorough treatment towards proving that both nonadaptive and adaptive composition can be reduced to this generalized random response at each step where some $t\in [0,\ep]$ is chosen either nonadaptively or adaptively.

Note that for composition over $\ep$-DP mechanisms, the worst-case mechanism is simply randomized response, hence there is no difference between the nonadaptive and adaptive setting because the worst-case is always randomized response regardless of previous outcomes. However, in our setting the same conclusion is not necessarily true because the adversary now has the power of adaptively choosing $t \in [0,\ep]$ at each round. We then first restrict our consideration to the nonadaptive setting and consider the optimal composition of this setting. 

\subsection{Nonadaptive optimal composition}

As with the previous work on advanced and optimal composition for $(\ep,\delta)$-DP mechanisms, it does not suffice to simply consider one pair $(\ep_g,\delta_g)$, but instead we consider a parameter as a function of the other parameter to get a full curve of privacy loss parameters.
Note that throughout this work, we will use similar conventions to  \cite{MurtaghVa16} in that $(\ep,\delta)$ will denote the privacy parameters of a single mechanism and $(\ep_g,\delta_g)$ will denote the global privacy parameters that are for the composition of these mechanisms.
While the previous optimal composition bounds considered fixing $\delta_g \in [0,1]$ and computing the optimal $\ep_g$ as a function of the $\delta_g$, it will be easier for us to write $\delta_g$ as an explicit formula of $\ep_g$ given Fact~\ref{fact:delta_opt}, which is also seen in  \cite{KairouzOhVi17,MurtaghVa16}. These formulations are equivalent, so for simplicity we will instead consider fixing $\ep_g$ and computing the optimal $\delta_g$. 

Having restricted our consideration to the nonadaptive setting and reducing to the worst-case mechanisms being our generalized random response, it is then straightforward to obtain the optimal composition formula for the heterogeneous setting of $\ep_i$-BR mechanisms.  We define the following class $\nonadaptBR^{1:k}$ of nonadaptive heterogeneous BR mechanisms and $\nonadaptBR^k$ of nonadaptive homogeneous BR mechanisms as 
\begin{equation}
\nonadaptBR^{1:k} \defeq  \{M_1 \times \cdots \times M_k : M_i \text{ is } \diffp_i\text{-BR}  \} \qquad \nonadaptBR^{k} \defeq  \{M_1 \times \cdots \times M_k : M_i \text{ is } \diffp\text{-BR}  \}.
\label{eq:nonadaptBR}
\end{equation}

\begin{restatable}[]{lemma}{nonadaptivehet}
\label{lem:non_interactive_GRR}
Recall from Definition~\ref{defn:gen_rr}
we have $p_{\diffp_i, t_i}, q_{\diffp_i,t_i}$.  We then have
\[
\delta_{\opt}(\nonadaptBR^{1:k},\diffp_g) = \sup_{\bbt \in \prod_{i \in [k]} [0,\ep_i]} \sum_{S \subseteq \{1,...,k\}} \max\left\{\prod_{i \notin S} q_{\diffp_i,t_i} \prod_{i \in S} (1 - q_{\diffp_i,t_i}) - e^{\diffp_g}\prod_{i \notin S}p_{\diffp_i,t_i} \prod_{i \in S} (1 - p_{\diffp_i,t_i}), 0 \right\}.
\]
\end{restatable}

Note this formulation can in some ways be seen as a generalization of the following result from \citet{MurtaghVa16}, although we only state it in the nonadaptive setting (as well as fix $\delta_i = 0$), it does also hold in the adaptive setting.
\begin{theorem}[Theorem 1.5 from  \citet{MurtaghVa16}]\label{thm:MVcomp}
Let $\cM_{\texttt{DP}}^{1:k}$ be the class of nonadaptive composed mechanism $M = M_1 \times \cdots \times M_k$ where each $M_i$ is $\diffp_i$-DP, then we have
\[
\delta_{\opt}(\cM_{\texttt{DP}}^{1:k},\diffp_g) = \frac{1}{\prod_{j=1}^k (1 + e^{\diffp_j} ) } \cdot  \sum_{S \subseteq \{1,...,k\}} \max\left\{\exp\left(\sum_{i \in S} \diffp_i \right) - e^{\diffp_g} \exp\left(\sum_{i \notin S} \diffp_i \right), 0 \right\}.
\]
\end{theorem} 
In particular, if we set $t_i = \frac{\ep_i}{2}$ for all $i$ instead of taking the sup, then this is equal to the LHS of the equation (1) in Theorem 1.5 for \cite{MurtaghVa16} where we replace $\ep_i$ with $\frac{\ep_i}{2}$. Equivalently, by setting $t_i = \frac{\ep_i}{2}$ for all $i$, this is the optimal composition of $\tfrac{\ep_1}{2}$-DP, $\ldots,\tfrac{\ep_k}{2}$-DP mechanisms.

Similar to the result in \citet{KairouzOhVi17} on optimal composition of DP mechanisms, we will restrict our consideration to the homogeneous setting where $\ep_1, \cdots, \ep_k = \ep$ in an attempt to obtain a formulation that is efficiently computable. However, this formulation will be far more difficult to simplify than the optimal composition of DP mechanisms because of the supremum over $t_i \in [0,\ep]$. Our simplification will require significant technical work that will ultimately be done in two key steps: 1) we show that the supremum is achieved when all $t_i = t_j$ for $i \neq j$, and 2) we show that the supremum is achieved by a certain value $t_i = t^* \in [0,\ep]$ contained in a set of at most $k$ possible values. This will then yield an explicit and efficiently computable formulation of the optimal nonadaptive composition of BR mechanisms.

\begin{restatable}[]{theorem}{nonadaptive}
\label{thm:non-adaptive}
Consider the homogeneous case where $\diffp_i = \diffp$ for each $i \in [k]$, then we have 
for $ p_{t_i} = p_{\diffp, t_i}$ given in Definition~\ref{defn:gen_rr} and setting $t_\ell^* = \frac{\diffp_g + (\ell + 1) \diffp}{k + 1}$ where if $t_{\ell}^* \notin [0,\ep]$, then we round it to the closest point in $[0,\ep]$
\[
\delta_{\opt}(\nonadaptBR^k,\diffp_g) = \max_{0\leqslant \ell \leqslant k} \sum_{i = 0}^k {k \choose i} p_{ t_{\ell}^* }^{k-i}(1 - p_{ t_{\ell}^* })^{i} \max\left\{\left( e^{k t_{\ell}^*  - i\diffp} - e^{\diffp_g} \right), 0 \right\}.
\]
	Furthermore, this can be computed in $O(k^2)$ time.
\end{restatable}

Once again, we note that by instead setting $t_{\ell}^* = \frac{\ep}{2}$, then this formulation is equivalent to the LHS of Theorem 1.4 in \cite{MurtaghVa16}, which is a rephrasing of the original optimal composition formulation in \cite{KairouzOhVi17}, where we replace $\ep$ with $\ep/2$.\footnote{Interestingly, this then implies that for any $\ep_g$ where this maximum is achieved with $t_{\ell} = \frac{\ep}{2}$, we then have that the optimal composition of $\ep$-BR mechanisms is equivalent to the optimal composition of $\frac{\ep}{2}$-DP mechanisms for that specific $\ep_g$. We have in fact tested this and found cases in which this is true, but could not find any discernible pattern for the specific values of $\ep_g$ when the optimal composition is equivalent.} 
We also plot the DP optimal composition bound where $\diffp/2$ is used as the DP privacy parameter for each individual mechanism in Figure~\ref{fig:ALL}.  
The improvement in this formulation over the optimal composition of $\diffp$-DP mechanisms is more than a factor of 2, and we empirically compare the bound for $\diffp_g$ in Figure~\ref{fig:ALL} as a function of $k$.  In the figure, we label ``DP OptComp" as the optimal composition bound for DP mechanisms from \cite{MurtaghVa16}, ``DR19" as the composition bound for $\diffp$-BR mechanisms from \cite{DurfeeRo19}, and ``BR OptComp" as the composition bound in Theorem~\ref{thm:non-adaptive}. 

Unfortunately, our proofs of this optimal composition formulation cannot be applied to the adaptive setting, pointing to the natural question of whether there is in fact further privacy loss when the adversary is given power to choose the mechanism based upon previous responses.

\subsection{Additional power of adaptivity\label{sect:added_power}}

In order to better explain the intuition behind optimal composition in both nonadaptive and adaptive settings, we rely upon the random walk interpretation of composition similar to analysis in \cite{DworkRoVa10,RogersRoUlVa16}. In particular, for composition of $\ep$-DP mechanisms, we can instead consider a random walk on the real line beginning at 0, where with probability $\frac{e^{\ep}}{e^{\ep} + 1}$ a step of $\ep$ is taken and with probability $\frac{1}{e^{\ep} + 1}$ a step of $-\ep$ is taken. Given some $\ep_g$, the goal of the adversary is to maximize the probability that the walk exceeds $\ep_g$ after $k$ steps and the amount in which it exceeds $\diffp_g$. For achieving an upper bound on the composition as in \cite{DworkRoVa10}, we can ignore the amount the walk exceeds $\ep_g$ and apply concentration bounds on the probability that the walk exceeds $\ep_g$ after $k$ steps. The optimal composition from \cite{KairouzOhVi17,MurtaghVa16} instead requires computing the resulting binomial distribution over the length of the walk to explicitly obtain both the probability and amount that each walk exceeds $\ep_g$.
In the nonadaptive setting, the reason we could also achieve an efficient formulation was because we proved that we can equivalently restrict all $t_i$ to be equal and further we can restrict the possible $t_i$ to a smaller set, so our computation once again became equivalent to examining each respective binomial distribution.

For the composition of $\ep$-DP mechanisms, the worst-case mechanism does not require an adversarial choice, however in our setting the adversary does have the power to choose each $t_i \in [0,\ep]$ in the generalized random response mechanism. This choice of $t_i$ will then exactly determine the length of the step in each direction, where either a step of $t_i$ is taken or a step of $t_i - \ep$ is taken. It might seem like the adversary would then always choose the maximum $t_i$, but the probability of taking that step is inversely related to the magnitude of the step. More specifically, the larger the adversary sets $t_i$, the smaller the probability that the step is taken in the positive direction, presenting a natural tradeoff. Following this random walk interpretation, we can then give an explicit optimal composition in the adaptive setting as a recursive formulation that incorporates the natural maximization problem.  

We begin by simplifying our notation for adaptive composition and focusing on the homogeneous case where $\diffp_i = \diffp$ for $i \in [k]$ and will address the heterogenous case in Section~\ref{sec:gap}. Given some fixed $\diffp>0$, let $\adaptBR^k \defeq (\mbr,\ldots,\mbr)$ be such that $\mbr$ is the class of $\ep$-BR mechanisms. We will denote the family of adaptive composition games over all adversaries as the following

\begin{equation}
\abr^k \defeq \{\AdaComp(\cA,\adaptBR^k,\cdot): \text{ adversary } \cA\}.
\label{eq:class_adaptive_comp}
\end{equation}

We then have the following result.  Note that we also consider the heterogenous case for $\diffp_1, \cdots, \diffp_k$ in Lemma~\ref{lem:adap_recursion_hetero}
\begin{lemma} 
Given global parameter $\diffp_g$ and $q_{\diffp, t_i}$ from  Definition~\ref{defn:gen_rr}, we have the following optimal privacy parameter where use set $\delta_{\opt}(\abr^0,\ep_g) = \max\{1 - e^{\ep_g}, 0 \}$,
\[
\delta_{\opt}(\abr^k,\ep_g) = \sup_{t_1 \in [0,\ep]} \left\{ q_{\ep,t_1} \delta_{\opt}(\abr^{k-1}, \ep_g - t_1) + (1 - q_{\ep,t_1})  \delta_{\opt}(\abr^{k-1}, \ep_g + \ep - t_1) \right\}.
\]
\end{lemma}

Note that this formulation does not necessarily hold for the nonadaptive setting because the choice of $t_2$ cannot change based upon the result of the first mechanism, and the supremum for all possible $t_i $ gets pulled to the beginning of the expression.
It is exactly this difference that will give the adversary additional power in the adaptive setting because, relying upon our random walk intuition, the natural tradeoff between the magnitude of the step $t_i$ and the probability of that step is actually dependent upon the current position of the random walk.
For example, consider a walk that begins by taking several steps in the negative direction.  In order to make up this increased distance and exceed $\ep_g$ it may then become necessary to increase the subsequent values of $t_i$ despite this decreasing probability of these steps occurring. Similarly, if the walk begins by taking several steps in the positive direction, it may become favorable to choose more conservative values of $t_i$ and increase the probability of taking these positive steps.

We rigorously confirm this intuition that will heavily rely upon having obtained an efficient formulation of the nonadaptive optimal composition. Furthermore, we confirm that this difference in privacy loss exists for all possible values of $k$ in our composition, and almost all choices of $\ep_g$. As would be expected, if $\ep_g = k\ep$ and basic composition can be applied, then there is no difference between optimal composition in the adaptive and nonadaptive setting. We further show that this slightly extends beyond just basic composition in which the adaptive and nonadaptive setting are equal, almost completely giving a full picture of when the adversary has additional power from adaptivity.

\begin{restatable}[]{theorem}{gap}
\label{thm:gap}
Recall the nonadaptive family of homogeneous $\diffp$-BR mechanisms $\nonadaptBR^k$ from \eqref{eq:nonadaptBR} and $\abr^k$ given in \eqref{eq:class_adaptive_comp}. For any $\ep_g \in [0,(k-3)\ep]$ we have,

\[
\delta_{\opt}(\abr^k,\ep_g) > \delta_{\opt}(\nonadaptBR^k,\ep_g).
\]

Further, for any $\ep_g \geq (k-1)\ep$, we have 

\[
\delta_{\opt}(\abr^k,\ep_g) = \delta_{\opt}(\nonadaptBR^k,\ep_g).
\]
\end{restatable}

Note that under these conditions the gap only exists for $k \geq 4$. We also show that the gap exists for $k = 2,3$ in Section~\ref{sec:gap}, but the conditions do not extend as nicely and we leave them out of the theorem statement here for simplicity.

We believe that the gap is quite small for all values of $\ep_g$ and $k$, however we believe that proving a strong upper bound on the gap would require significant technical work and leave it to future work. We can confirm this numerically for reasonable $k$, but due to the nature of the recursive formulation for the adaptive setting it is intractable to check this for larger values of $k$. Furthermore, these numerical methods become even more computationally difficult for the case of heterogenous privacy parameters and the gap for this setting may be much larger.

\subsection{Improved and efficiently computable bounds for adaptive composition}

While we gave an explicit formulation of the optimal composition for the adaptive setting of BR mechanisms, the computation is not tractable, and we suspect that it has similar hardness results to  \cite{MurtaghVa16}, which we leave to future work. 
Accordingly, we further improve the known efficiently computable upper bounds on the adaptive composition of $\ep$-BR mechanisms from \cite{DurfeeRo19}. 
The previous work on $\ep$-BR composition followed a similar approach to \cite{DworkRoVa10} applying both an Azuma-Hoeffding bound (on the variance) and a KL divergence bound (on the bias) to achieve a reasonably simple upper bound on the optimal composition. 
However, the previous work only considered using the BR property to improve the bound from Azuma-Hoeffding and did not consider improving the KL divergence bound.
While these bounds are quite complex to generally compute, we note that for our generalized random response it will actually be quite simple to compute the explicit KL divergence. Using our reduction to this worst-case class of mechanisms and taking the supremum over all choices of $t$ we can give a much improved bound on the KL divergence. 

\begin{restatable}[]{corollary}{mgf}
\label{cor:OptKL}
Let $\vcM \defeq (\cM_1, \cM_2, \cdots, \cM_k)$ where each $\cM_i$ is the class of $\diffp_i$-BR mechanisms.  We then have that $\vcM$ is $(\diffp_g(\delta_g), \delta_g)$-DP under $k$-fold adaptive composition for any $\delta_g \geq 0$ where
\[
\diffp_g(\delta) = \min\left\{ \sum_{i=1}^k \diffp_i,\sum_{i=1}^k \left( \frac{\diffp_i}{1 - e^{-\diffp_i}} - 1 - \log \left( \frac{\diffp_i}{1 - e^{-\diffp_i}} \right) \right)+  \sqrt{ \frac{1}{2} \sum_{i=1}^k \diffp_i^2 \log(1/\delta)}\right\}. 
\]
\end{restatable}

This gives substantial improvements over the previous bound in some settings (and we will provide plots in Section~\ref{sec:mgf_approach}), but we will further improve this bound. In particular, the bound given above considers the KL divergence and Azuma-Hoeffding separately, which is to say that the worst-case $t_i \in [0,\ep_i]$ is chosen separately for these two bounds instead of choosing this supremum with respect to both. In order to improve this, we backtrack a step in this method and use the same techniques from the proof of Azuma-Hoeffding but apply our more exact characterization.

\begin{restatable}[]{theorem}{mgfnew}
\label{thm:MGFthm}
Let $\vcM \defeq (\cM_1, \cM_2, \cdots, \cM_k)$ each $\cM_i$ is the class of $\diffp_i$-BR mechanisms.  We then have that $\vcM$ is $(\diffp_g, \delta_g(\diffp_g))$-DP under $k$-fold adaptive composition for any $\diffp_g \geq 0$ where we define $h_\diffp(\lambda):= \sup_{t\in[0,\diffp]}\lambda(\diffp-t)+\log\big(1+p_{\diffp,t}(e^{-\lambda\diffp}-1)\big)$ with $p_{\diffp,t} = \frac{e^{-t}-e^{-\diffp}}{1-e^{-\diffp}}$ and
	\[
	\delta_g(\ep_g) = \inf_{\lambda>0}\e^{-\lambda\ep_g+ \sum_i h_{\diffp_i}(\lambda)}.
	\]
\end{restatable}

We present plots of our results in Figure~\ref{fig:ALL} for the homogeneous case, plotting $\diffp_g$ as a function of $k$.  As stated earlier, we label ``$\diffp$-DP OptComp" as the optimal composition bound for DP mechanisms from \cite{MurtaghVa16}, ``DR19" as the composition bound for $\diffp$-BR mechanisms from \cite{DurfeeRo19}, and ``BR OptComp" as the composition bound in Theorem~\ref{thm:non-adaptive}, which only applies in the nonadaptive setting. Furthermore, we label ``OptKL" as the bound from Corollary~\ref{cor:OptKL} and ``MGF" as the bound in Theorem~\ref{thm:MGFthm}.  To compare our bounds with simply using the optimal DP composition bound with a half the actual privacy parameter, we also plot the DP optimal composition bound with $\diffp/2$ with label ``$\diffp/2$-DP OptComp".   This last curve highlights the fact that $\diffp$-BR is almost the same as $\diffp/2$-DP when applying composition.

\begin{figure}[h]
\centering
\includegraphics[width=0.35\textwidth]{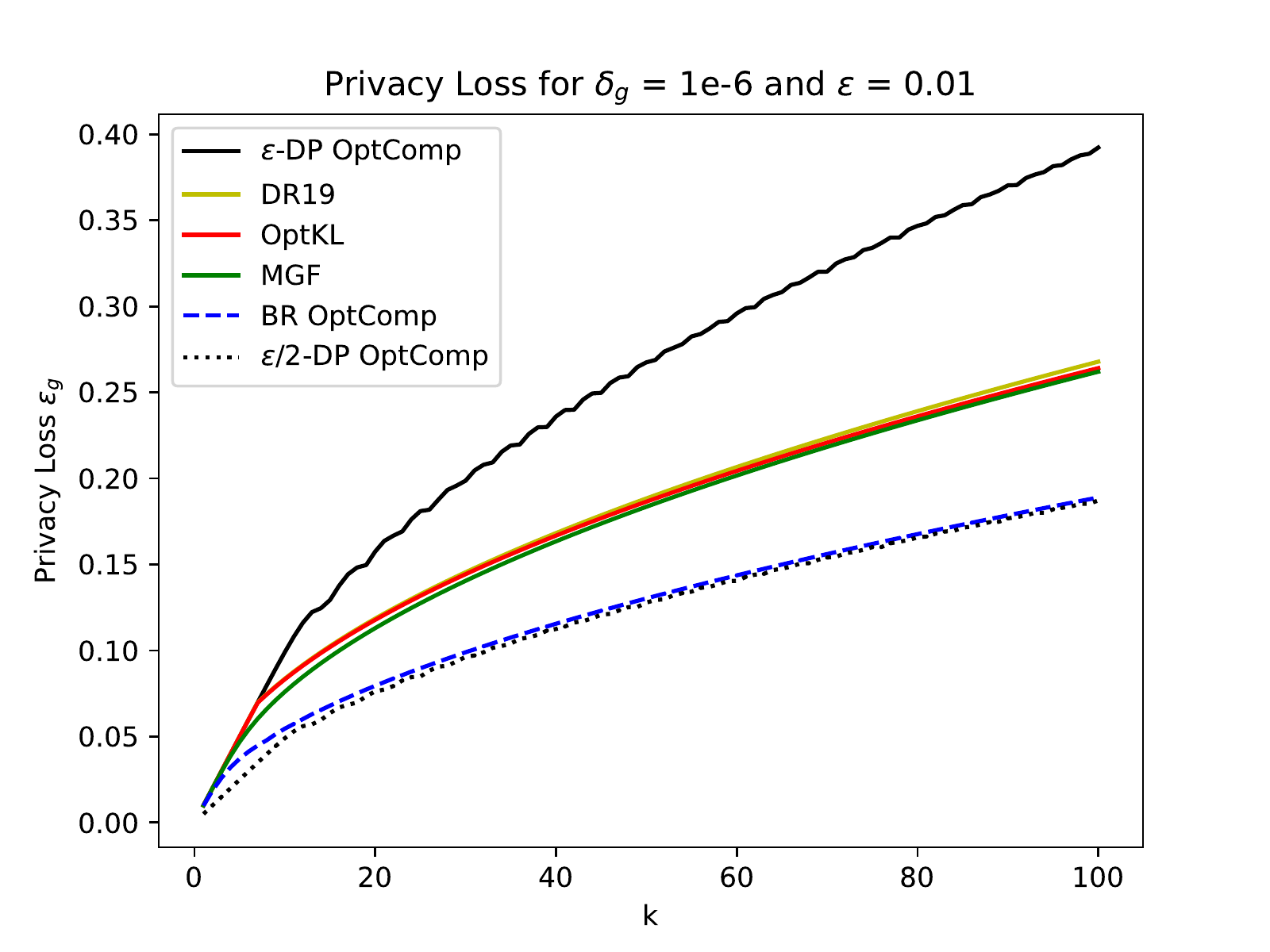}
\includegraphics[width=0.35\textwidth]{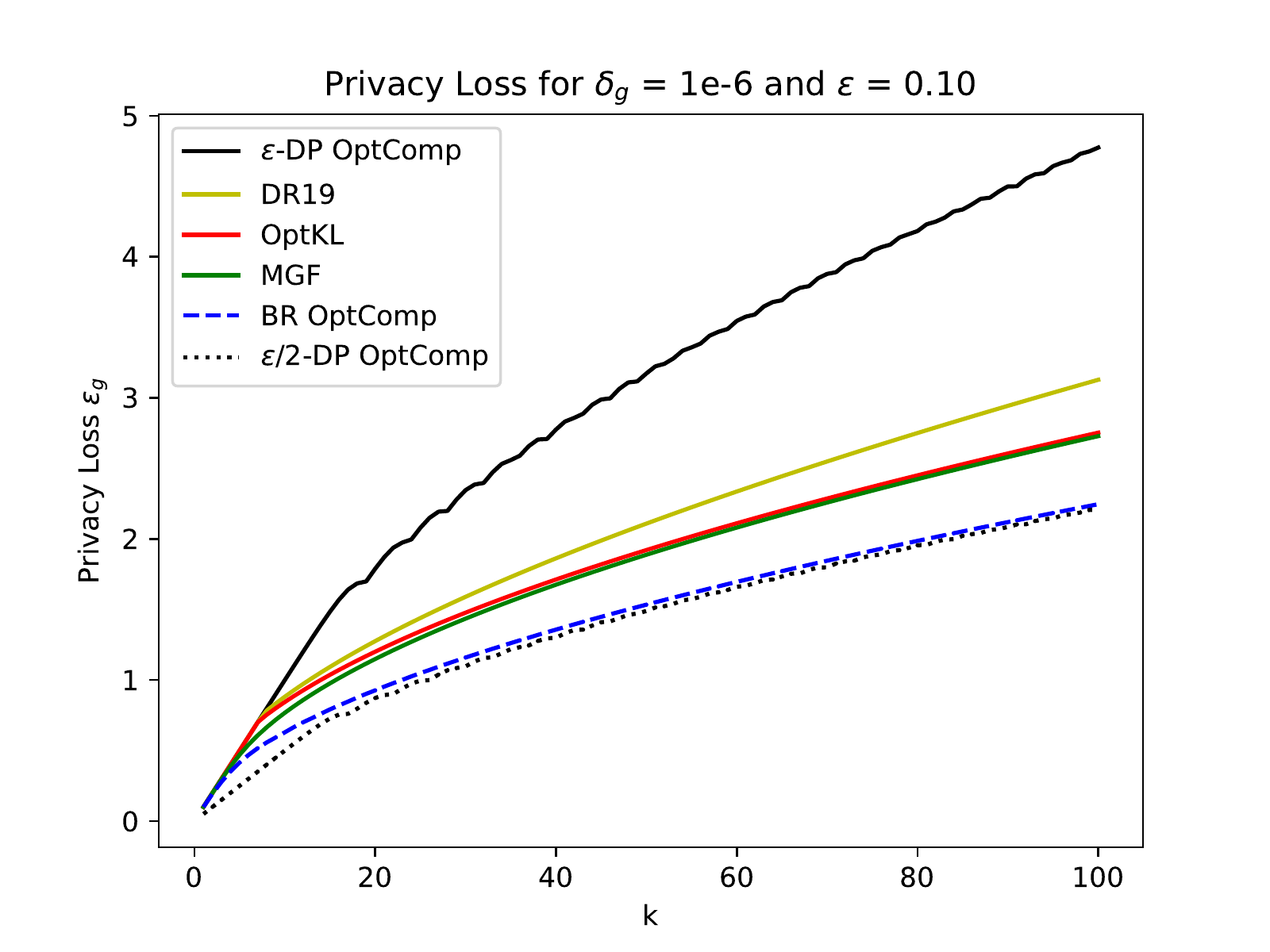}
\includegraphics[width=0.35\textwidth]{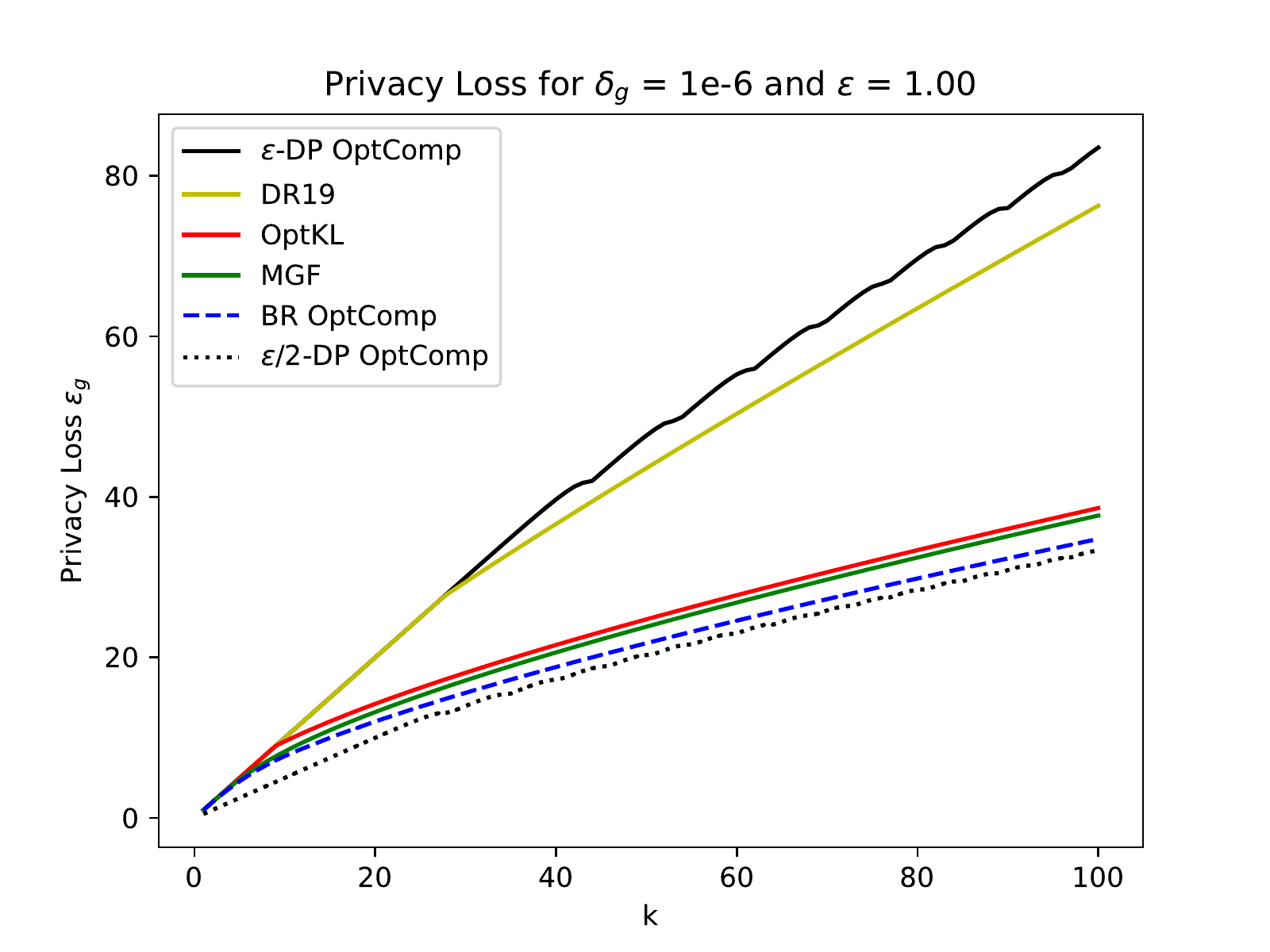}
\caption{Comparison of optimal DP composition with the BR composition bounds in this work and in \citet{DurfeeRo19}.  The dashed curve only applies in the nonadaptive composition setting and the dotted curve uses the existing DP optimal composition bound with half the actual privacy parameter.  We present results for $\delta_g = 10^{-6}$ and $\diffp \in \{ 0.01, 0.1, 1\}$.\label{fig:ALL}}
\end{figure}

\subsection{Discussion of optimal DP composition bounds \label{sec:OptCompVsBR}}
Although $\diffp$-BR implies $\diffp$-DP, and the converse holds up to a factor of 2 in the privacy parameter, it is important to point out that our optimal composition analysis of BR mechanisms does not supersede the optimal composition of DP mechanisms.  More specifically, consider the Laplace mechanism \cite{DworkMcNiSm06}, which adds Laplace noise to a bounded sensitivity statistic.  This mechanism is $\diffp$-DP, but it is also $2 \diffp$-BR yet it has a fixed value $t = \diffp$ for any neighboring datasets.  As we will discuss more rigorously in our analysis, our optimal composition bounds for BR mechanisms follows from maximizing the bound over all sequences of $t$ values.  Hence, utilizing the optimal composition bound over BR mechanisms will result in a larger than necessary bound when considering Laplace mechanisms, and thus the optimal DP composition bounds from \cite{KairouzOhVi17,MurtaghVa16} should be used.  Alternatively, if we are composing exponential mechanisms that we know are $\diffp$-BR, then our composition bounds improves on the optimal composition of $\diffp$-DP mechanisms.

Consider the following example with randomized response.  In this case $M_{\texttt{RR}}: \{0,1\} \to \{0,1\}$ and $M_{\texttt{RR}}(b;\diffp) = b$ with probability $\tfrac{e^\diffp}{e^\diffp + 1}$.  To fit this into the generic exponential mechanism, we require a quality score $u(b,b')$ and we need to calculate its sensitivity, or as we discussed in Proposition~\ref{prop:semantics}, its range.  In this case $u(b,b') = \1{b = b'}$, which has sensitivity$\Delta u = 1$ and also has range $\tilde\Delta u = 2$.  Whether we use the range or the sensitivity of the quality score, the generic exponential mechanism is then written as $M_u(b; \diffp) = \tfrac{e^{\diffp q(b,b)/2}}{e^{\diffp q(b,b)/2} + e^{\diffp q(b,1-b)/2}} = \tfrac{e^{\diffp/2}}{e^{\diffp/2} + 1}$.  Hence, we have $M_u(\cdot;2 \diffp) = M_{\texttt{RR}}(\cdot;\diffp)$.  The fact that randomized response can be written as $M_u(\cdot;2 \diffp)$ implies that it is $2 \diffp$-BR, but we further note that there are only two neighboring databases for randomized response. This then allows for only one value $t \in [0,2\diffp]$ from Corollary~\ref{cor:br}, where we see that $t = \diffp$ implies that this randomized response is also $\diffp$-DP.
Accordingly, if we only knew the generic exponential form with parameter $\diffp$ then our composition bounds would improve over the general optimal DP composition bounds from \cite{MurtaghVa16,KairouzOhVi17}.  However, if it is also known that each individual mechanism is also $\diffp/2$-DP, as is the case for randomized response with parameter $\diffp/2$, then the bounds from \cite{MurtaghVa16,KairouzOhVi17} cannot be improved.

%% file: reduction.tex
\section{Bounded range and generalized random response}\label{sec:reduction} 

In this section, we show that $k$-fold adaptive composition over the class of BR mechanisms can be reduced to only considering adversaries that select a generalized random response mechanism at each step. First, we show that we can post-process the generalized random response to simulate any BR mechanism on neighboring inputs. For this proof, we will utilize the hypothesis testing interpretation of DP that was similarly used in \cite{KairouzOhVi17} and then extended in  \cite{DongRoSu19}. We defer the analysis to Appendix~\ref{sec:gen-rand-response}.

\begin{lemma}\label{lem:reduc}
Let mechanism $M: \cX \to \cY$ be $\diffp$-BR. For any neighboring databases $x^0,x^1 \in \cX$, there exists some $t=t(M,x^0,x^1) \in [0,\diffp]$ and randomized function $\phi: \{0,1\} \to \cY$ that depends on $M,x^0,x^1$ such that for any $y \in Y$ and $b \in \{0,1\}$ we have the following equivalence in terms of the generalized randomized response mechanism from Definition~\ref{defn:gen_rr}.

\[
\Pr[M(x^b) = y] = \Pr[\phi(\grr{\ep,t}(b)) = y]
\]
\label{lem:pp_rr}
\end{lemma}

We next show that $k$-fold adaptive composition over the class of BR mechanisms is equivalent to considering the class of generalized randomized response mechanisms instead.

\begin{lemma}
Fix parameters $\diffp_1, \cdots, \diffp_k$.  Let $\vcM = (\cM_1, \cdots, \cM_k)$ be such that $\cM_i$ is the class of $\diffp_i$-BR mechanisms, and let $\vcRR = (\cRR_1, \cdots, \cRR_k)$ be the class such that $\cRR_i \defeq \{ \grr{\diffp_i,t_i} : t_i \in [0,\diffp_i] \}$.  We then have that $\vcM$ is $(\diffp_g,\delta_g)$-DP under $k$-fold adaptive composition if and only if $\vcRR$ is $(\diffp_g,\delta_g)$-DP under $k$-fold adaptive composition.
\label{lem:lem1}
\end{lemma}

\begin{proof}

Take any $\cA = (\cDe,\cR)$ that selects mechanisms from $\cM_i$ at round $i$ and we will construct $\cA' = (\cDe',\cR')$ that selects mechanisms in $\cRR_i$ in the following way.  Replace the deterministic component $\cDe(r_0,A_1^b, \cdots r_{i-1})$ that selects neighbors $x^{0}_i, x^{1}_i$ and $M_i \in \cM_i$ at each round $i$ with $\cDe'(r_0,B_1^b,A_1^b, \cdots r_{i-1})$ that selects neighbors $x^{0}_i, x^{1}_i$ and $t_i(M_i,x^{0}_i, x^{1}_i)$ where $B_\ell^b = \grr{t(\cM_\ell,x^0_\ell,x^{1}_\ell)}(b)$ and $\ell < i$.

The new analyst $\cA'$ receives $B_i^b = \grr{t(\cM_i,x^{0}_i,x^{1}_i)}(b)$ whereas $\cA$ receives $A_i^b = M_i(x_i^b)$.  We then construct the randomized component of $\cA'$ in the following way.  Rather than sample $r_i \sim \cR(r_1,A_1^b, \cdots r_{i-1},A_{i-1}^b)$, we sample $r_i' = (A_{i}^b , r_i) \sim \cR'(r_1,B_1^b,A_1^b, \cdots, r_{i-1},B_{i}^b)$ where first $A_{i}^b = \phi_{i}(B_{i}^b)$ and $\phi_i$ is the post-processing function described in Lemma~\ref{lem:pp_rr} that depends on $\cM_i,x^{0}_i,x^{1}_i$, then $r_i \sim \cR(r_1,A_1^b, \cdots r_{i-1},A_{i-1}^b)$, as before.

Given any outcome $(r_0,A_1^b, \cdots r_{k-1},A_k^b,r_k)$, we know that there exists a post-processing function $\psi$ such that for $b \in \{0,1\}$

\begin{multline*}
\Pr\left[\AdaComp(\cA,\vcM,b) = \left(r_0,A_1^b, \cdots r_{k-1},A_k^b,r_k\right)\right] \\
= \Pr\left[\psi\left( \AdaComp(\cA',\vcRR,b)  \right) = \left(r_0, A_1^b, \cdots r_{k-1}, A_k^b,r_k\right) \right]
\end{multline*}
\end{proof}

\subsection{Handling convexity for BR composition\label{sec:convex_combo}}

In this section, we discuss a technicality for adaptive composition of BR mechanisms.  As discussed earlier, BR mechanisms are not closed under convex combinations, and this can be easily seen by simply considering a mechanism that has four possible outputs from randomizing over $\grr{\diffp,t_1}$ and $\grr{\diffp,t_2}$ where $t_1 \neq t_2$. This allows adversaries potentially additional power when they can randomize between different BR mechanisms at each round, which is not the case for classes of mechanisms that are closed under convex combinations, such as DP. 

Despite this technicality, we will show that allowing the analyst this adaptive randomness at each step does not increase the privacy loss.  
Consider the same adaptive game in Algorithm~\ref{algo:adaptgame}, but now we take away the adversary's ability to add their own data-independent randomness at each round, which we will denote as $\cA = (\emptyset, \cDe)$. We will show that this has the same level of privacy regardless of the class of randomized algorithms used.  

\begin{definition}[Adaptive Composition without Adversarial Randomness] 
Given classes of randomized algorithms $\vcE = (\cE_1,\cdots \cE_k)$, we say $\vcE$ is $(\diffp_g,\delta_g)$ differentially private under $k$-fold adaptive composition without adversarial randomness if for any adversary $\cA = (\emptyset, \cDe)$ that does not have any randomness of its own and $b \in \{0,1\}$, along with any set $S$ that is a subset of outputs of $\AdaComp((\emptyset,\cDe),\cE,\cdot)$

\[
\Pr[\AdaComp((\emptyset,\cDe),\vcE,b) \in S] \leq e^{\diffp_g}\Pr[\AdaComp((\emptyset,\cDe),\vcE,1-b) \in S] + \delta_g
\]

\end{definition}

\begin{lemma}
Given any class of randomized algorithms  $\vcE = (\cE_1,\cdots, \cE_k)$, $\vcE$ is $(\diffp_g,\delta_g)$-DP under $k$-fold adaptive composition without adversarial randomness if and only if $\vcE$ is $(\diffp_g,\delta_g)$-DP under $k$-fold adaptive composition.
\label{lem:lem2}
\end{lemma}

\begin{proof}
We largely follow Lemma 3.4 in \citet{RogersRUV16arxiv} which shows the point-wise equivalence between an adversary that has access to internal randomness and with a deterministic adversary who can then post-processes the final result.  This is done by including \emph{simulated} randomness for the deterministic adversary that can be fixed prior to any interaction with the dataset.  One technical difference between our setting and theirs is that for them an adversary can select a DP algorithm, which is then a post-processing function of randomized response, at each round.  This means that even if an adversary could additionally randomize between different DP algorithms at each round, the result is still DP.  In our case, there is a difference between a deterministic adversary and an adversary that can randomize between BR mechanisms at each round, because the resulting mechanism may no longer be BR.  However, we can just include this internal randomness of the adversary at each round in the simulated randomness from the analysis in Lemma 3.4 of \cite{RogersRUV16arxiv}.  Hence, we can analyze the DP guarantees for each realized value of simulated randomness.  Lastly, the DP guarantee does not change under convex combinations of the realized simulated randomness, which shows that it suffices to only consider deterministic adversaries. 
\end{proof}

Using Lemmas~\ref{lem:lem1} and \ref{lem:lem2}, we have the immediate result which shows that without loss of generality, we can consider deterministic adversaries that can select generalized randomized response mechanisms at each round.

\begin{corollary}\label{cor:reduc}
Fix parameters $\diffp_1, \cdots, \diffp_k$.  Let $\vcM = (\cM_1, \cdots, \cM_k)$ be such that $\cM_i$ is the class of $\diffp_i$-BR mechanisms, and let $\vcRR = (\cRR_1, \cdots, \cRR_k)$ be the class such that $\cRR_i \defeq \{ \grr{\diffp_i,t_i} : t_i \in [0,\diffp_i] \}$.  We then have that $\vcM$ is $(\diffp_g,\delta_g)$-DP under $k$-fold adaptive composition if and only if $\vcRR$ is $(\diffp_g,\delta_g)$-DP under $k$-fold adaptive composition without adversarial randomness.
\end{corollary}

\subsection{Exponential mechanism equivalence to generalized random response\label{sec:exp_to_grr}}
It was shown in \citet{KairouzOhVi17} that the discretized version of the Laplace mechanism, i.e. the geometric mechanism, has the largest privacy degradation under composition.  Similarly, we show that for certain quality scores the exponential mechanism is equal in distribution, up to a data independent post processing function, as the generalized randomized response mechanism.  Among this class of quality scores is the commonly used score for counting queries.  
More specifically, if we run an exponential mechanism, then by post-processing we can achieve the same distribution as $\grr{\diffp,t}$ for some $t$, and likewise if we run $\grr{\diffp,t}$ with the same $t$, then by post-processing we can achieve the same distribution as the exponential mechanism. We first define the exponential mechanism that we will be considering. This mechanism is one of the most common uses of the exponential mechanism where each individual's data is a bit string over some domain, and the mechanism wants to output the maximum count for all individuals over this domain.

\begin{definition}
Let $\cX \equiv \{0,1 \}^{n \times d}$ and $\bbx = (x_{i,j} : i \in [n], j \in [d]) \in \cX$ for some $n \in \mathbb{N}$, and define $M_{CQ}: \cX \to [d]$ to be the $\diffp$-DP exponential mechanism from Definition~\ref{defn:em} with quality score $u(x,j) = \sum_{i =1}^n x_{i,j}$. Neighboring databases will result from the addition or subtraction of a bit string $x_i = \{0,1\}^d$. Note here that $\Delta u = 1$ and that $u$ is also monotonic.

\end{definition}

Similar to the generalized random response mechanism, we then show that for any neighboring databases the log-ratio of the probability mass for any outcome $j \in [d]$ is only at the end points of the range.

\begin{lemma}\label{lem:cq_extremes}
For any neighboring databases $\bbx,\bbx' \in \cX$ there exists some $t \in [0,\ep]$ such that for any outcome $j \in [d]$

\[
\log\left( \frac{\Pr[M_{CQ}(\bbx) = j]}{\Pr[M_{CQ}(\bbx') = j]} \right) \in \left\{t-\ep,t\right\}
\]

\end{lemma}

\begin{proof}
We first assume that $\bbx' = \bbx + x_i$ where $x_i \in \{0,1\}^d$. We first set 

\[
t = \log \left( \frac{\sum_{j\in [d]} e^{\ep u(\bbx',j)}}{\sum_{j\in [d]} e^{\ep u(\bbx,j)}} \right)
\]

Note that we must have $t \in [0,\ep]$ because $u(\bbx,j) + 1 \geq u(\bbx',j) \geq u(\bbx,j)$ for all $j \in [d]$. We can then reduce our probability log-ratio to

\[
\log\left( \frac{\Pr[M_{CQ}(\bbx) = j]}{\Pr[M_{CQ}(\bbx') = j]} \right) = t + \log \left(\frac{e^{\ep u(\bbx,j)}}{e^{\ep u(\bbx',j)}} \right)
\]

Applying our assumption that $\bbx' = \bbx + x_i$, by the definition of $u$ we have $u(\bbx',j) = u(\bbx,j) + x_{i,j}$, which reduces our expression to

\[
\log\left( \frac{\Pr[M_{CQ}(\bbx) = j]}{\Pr[M_{CQ}(\bbx') = j]} \right) = t - \ep x_{i,j}
\]
and this implies our desired result because $x_{i,j} \in \{0,1\}$. We assumed $\bbx' = \bbx + x_i$ and considering the other case is equivalent to flipping the fraction, where it follows from natural log properties that

\[
\log\left( \frac{\Pr[M_{CQ}(\bbx') = j]}{\Pr[M_{CQ}(\bbx) = j]} \right) = \ep x_{i,j} - t
\]
which also implies our desired result because $\ep - t \in [0,\ep]$.

\end{proof}

This result is exactly why we consider the relation between this mechanism and generalized random response to be analogous to the relation between geometric noise and randomized response. For any outcome in the geometric mechanism, the magnitude of the log-ratio is always $\ep$, but unlike randomized response there are many more than two possible outcomes. Essentially, we can consider geometric noise and this counting query mechanism to split the outcomes of their respective randomized response into many outcomes, which will be the post-processing function.

\begin{corollary}
For any neighboring databases $\bbx^0,\bbx^1$ then there must exist some $t \in [0,\ep]$ and post-processing functions $\phi$ and $\phi'$ such that $M_{CQ}(\bbx^b) \equiv \phi(\grr{\ep,t}(b))$ and $\phi'(M_{CQ}(\bbx^b)) \equiv \grr{\ep,t}(b)$

\end{corollary}

\begin{proof}
Applying Lemma~\ref{lem:cq_extremes}, we split the outcome indices in the following way with $b' \in \{0,1\}$

\[
\cI_{b'} = \left\{j \in [d] : \log\left( \frac{\Pr[M_{CQ}(\bbx^0) = j]}{\Pr[M_{CQ}(\bbx^1) = j]} \right) = t - \ep b' \right\}.
\]
It is straightforward to see from Definition~\ref{defn:gen_rr} that we also have 

\[
\log\left( \frac{\Pr[\grr{\ep,t}(0) = b']}{\Pr[\grr{\ep,t}(1) = b']} \right) = t - \ep b'.
\]
Therefore, we must have for any $b \in \{0,1\}$ and $b' \in \{0,1\}$ that

\[
\Pr[\grr{\ep,t}(b) = b'] = \sum_{j \in \cI_{b'}} \Pr[M_{CQ}(\bbx^b) = j]
\]
and our claim follows easily.
\end{proof}

From Corollary~\ref{cor:reduc}, we know that the adaptive composition of BR mechanisms can be reduced to the class of generalized random responses and that this class is parameterized over all $t \in [0,\ep]$. 
In our proof of Lemma~\ref{lem:cq_extremes} we showed that the value $t$ came from the log-ratio of the sum of exponential functions.
For our definition of $\cX$, the number of neighboring databases is countably infinite, so it is technically impossible for there to always exist some neighboring databases with a corresponding $t$ over the uncountably infinite interval $[0,\ep]$. However, we can find neighboring databases that give a log-ratio arbitrarily close to any given $t \in [0,\ep]$,
i.e. the set of possible $t$ values from neighboring databases is dense in $[0,\ep]$, and for all practical purposes we can consider them equivalent.
Therefore, the adaptive composition game with this simple instantiation of the exponential mechanism is equivalent to an adversary being restricted to the class of generalized randomized response mechanisms at each round.  This is comparable to the result in \citet{KairouzOhVi17} that shows that the geometric mechanism achieves the worst case privacy composition bound since it also achieves the same privacy region as the standard randomized response once the neighboring datasets are fixed at each round.

%% file: non-interactive.tex
\section{Nonadaptive optimal composition}\label{sec:nonadaptive} 

In this section, we first give the explicit formulation for the optimal composition of nonadaptive BR mechanisms originally stated in Lemma~\ref{lem:non_interactive_GRR}. The majority of the section will then be devoted to reducing this formulation to a simpler formula that can be computed in $O(k^2)$ time for the homogeneous composition case, i.e. all privacy parameters are the same at each round. This will then culminate in a proof of Theorem~\ref{thm:non-adaptive}.

 We will denote $\bbt = (t_1, \cdots, t_k) \in \prod [0,\ep_i]$ where $\prod [0,\ep_i] \defeq [0,\diffp_1] \times \cdots \times [0,\diffp_k]$ and if all $\diffp_i = \diffp$ we will simply write $[0,\ep]^k$.  Recall from \eqref{eq:nonadaptBR}, we will denote the family of nonadaptive BR mechanisms as $\nonadaptBR^k$ for the homogeneous case and $\nonadaptBR^{1:k}$ for the heterogeneous case.  Recall that we defined the optimal privacy parameters by fixing a global $\diffp_g$ and giving a formula for $\delta_\opt$ in terms of $\diffp_g$ as in \eqref{eq:opt_delta}.  Our first formulation follows immediately from Lemma~\ref{lem:pp_rr}. 

\begin{lemma}\label{lem:reduction_to_grr}
\begin{align*}
& \delta_{\opt}(\nonadaptBR^{1:k},\diffp_g)\\
& \qquad = \sup_{\bbt \in \prod [0,\ep_i]} \max_{\bbb \in \{0,1\}^k} \sum_{\bby \in \{0,1\}^k} \max \left\{ \prod_{i=1}^k \Pr[\grr{\diffp_i,t_i}(b_i) = y_i] - e^{\diffp_g} \prod_{i=1}^k \Pr[\grr{\diffp_i,t_i}(1 - b_i) = y_i], 0 \right\}.
\end{align*}
\end{lemma}

\begin{proof}
We know that DP is closed under post-processing, so from Lemma~\ref{lem:pp_rr} we can restrict our consideration to $\grr{\diffp_i,t_i}$ for $t_i \in [0,\diffp_i]$, along with $b_i \in \{0,1\}$. The formulation then follows from Definition~\ref{def:optimal_delta} and Fact~\ref{fact:delta_opt}.
\end{proof}

We have the following symmetry result for the generalized randomized response mechanism, which will be useful in our analysis.

\begin{claim}\label{claim:symmetric}
For any $b \in \{0,1\}$ along with $\diffp \geq 0$ and $t \in [0,\diffp]$ we have
\[
\Pr[\grr{\diffp,t}(b) = b] = \Pr[\grr{\diffp,\diffp-t}(1-b) = 1-b].
\]
\end{claim}

We then use this symmetry property to show that the choice of $b_i$ is irrelevant.

\begin{corollary}\label{cor:symmetry_grr}
For any $\bbt \in \prod [0,\ep_i]$ and $\bbb \in \{0,1\}^k$,  and some fixed $b \in \{0,1\}$, there exists $\bbt' \in \prod [0,\ep_i]$ such that 
\begin{multline*}
\sum_{\bby \in \{0,1\}^k} \max \left\{ \prod_{i=1}^k \Pr[\grr{\diffp_i,t_i}(b_i) = y_i] - e^{\diffp_g} \prod_{i=1}^k \Pr[\grr{\diffp_i,t_i}(1 - b_i) = y_i], 0 \right\} \\
=
\sum_{\bby \in \{0,1\}^k} \max \left\{ \prod_{i=1}^k \Pr[\grr{\diffp_i,t'_i}(b) = y_i] - e^{\diffp_g} \prod_{i=1}^k \Pr[\grr{\diffp_i,t'_i}(1 - b) = y_i], 0 \right\} 
\end{multline*}

\end{corollary}

\begin{proof}

If $b_i = b$, then we can simply set $t'_i = t_i$. If $b_i \neq b$, then from Claim~\ref{claim:symmetric} we can set $t'_i = \ep_i - t_i$ and the value of the summation will not change.
\end{proof}

It then follows that we can fix $b\in \{0,1\}$ to give a simpler expression, and this expression is also a generalization of the optimal composition bound in Theorem~\ref{thm:MVcomp}, where instead of the $\sup$ term over $\bbt \in \prod [0,\ep_i]$, we can set each $t_i = \diffp_i/2$, and this becomes the optimal composition of $\frac{\ep_i}{2}$-DP mechanisms.

\nonadaptivehet*

\begin{proof}
Follows immediately from applying Corollary~\ref{cor:symmetry_grr} with $b = 0$ to Lemma~\ref{lem:reduction_to_grr}.
\end{proof}

\subsection{Simplifying the optimal composition bound for the homogeneous case}
Although we have a formula for the optimal composition bound over BR mechanisms, it is intractable to compute for even modest values of $k$.  To help simplify things, we will now restrict our consideration to the homogeneous case, where all $\diffp_i = \diffp \geq 0$, and we will drop the $\diffp$ from our notation, e.g. $p_{\diffp,t_i} \equiv p_{t_i}$.  We conjecture that the heterogeneous case has a similar hardness result to compute as the result in \citet{MurtaghVa16}, but we leave that as an open problem.  

Since we have shown that $\delta_\opt(\nonadaptBR^k,\diffp_g)$ can be written as a $\sup$ over $\bbt \in [0,\ep]^k$, we will define the function $\delta: [0,\ep]^k \times \R \to [0,1]$ as the following

\begin{equation}
\delta(\bbt,\diffp_g) \defeq \sum_{S \subseteq \{1,...,k\}} \max\left\{\prod_{i \notin S} q_{t_i} \prod_{i \in S} (1 - q_{t_i}) - e^{\diffp_g}\prod_{i \notin S}p_{t_i} \prod_{i \in S} (1 - p_{t_i}), 0 \right\}.
\label{eq:delta_t}
\end{equation}
Written in this way, we have $\delta_\opt(\nonadaptBR^k,\diffp_g) = \sup_{\bbt \in [0,\ep]^k} \delta(\bbt,\diffp_g)$.  We first show that when $\diffp_g\notin (-k\diffp,k\diffp)$, then the choice of $\delta(\bbt,\diffp_g)$ does not depend on $\bbt \in [0,\ep]^k$.  However, this region for $\diffp_g$ is not typically interesting in most DP applications, since $\diffp_g = k\diffp$ is simply applying \emph{basic composition} from \citet{DworkMcNiSm06}.

\begin{lemma}\label{lem:ep_g_edge_case}
For any $\bbt \in [0,\ep]^k$, if $\diffp_g \leq -k\diffp$ then $\delta(\bbt,\diffp_g) = 1 - e^{\diffp_g}$, and if $\diffp_g \geq k\diffp$ then $\delta(\bbt,\diffp_g) = 0$.
\end{lemma}

\begin{proof}
Using the fact that $q_{t} = e^t p_t$ and $(1 - q_t) = e^{t-\diffp} (1 - p_t)$, we equivalently have 

\[
\delta(\bbt,\diffp_g)  = \sum_{S \subseteq \{1,...,k\}} \prod_{i \notin S}p_{t_i} \prod_{i \in S} (1 - p_{t_i}) \max\left\{e^{\sum t_i - |S|\diffp} - e^{\diffp_g}, 0 \right\}
\]

If $\diffp_g \geq k\diffp$ then $\max\{e^{\sum t_i - |S|\diffp} - e^{\diffp_g}, 0 \} = 0$ for any $S \subseteq \{1,\ldots,k\}$. Similarly, if $\diffp_g \leq -k\diffp$ then $\max\{e^{\sum t_i - |S|\diffp} - e^{\diffp_g}, 0 \} = e^{\sum t_i - |S|\diffp} - e^{\diffp_g}$ for any $S \subseteq \{1,\ldots,k\}$ and we get 
\[
\delta(\bbt,\diffp_g)  = \sum_{S \subseteq \{1,...,k\}} \left( \prod_{i \notin S} q_{t_i} \prod_{i \in S} (1 - q_{t_i}) - e^{\diffp_g}\prod_{i \notin S}p_{t_i} \prod_{i \in S} (1 - p_{t_i})\right) = 1 - e^{\diffp_g} 
\]
\end{proof}

For the remainder of our analysis, we will focus on the interesting setting where $\diffp_g \in (-k\diffp,k\diffp)$.  Despite the large domain $[0,\ep]^k$ of values to choose from in the $\sup_\bbt$ for $\delta_\opt$, we show that it suffices to consider the much smaller domain where each $t_i = t^*$ for some $t^*$ for each $i \in [k]$.  This result is crucial in determining a formula that can be computed efficiently for $\delta_\opt$.  
We first give an easy condition on what the $t_i$ must satisfy to optimize the $\delta$ parameter which will be important for proving a strict inequality in the subsequent claim. 

\begin{lemma}\label{lem:bound_sum_t}
If $\diffp_g \in (-k\diffp,k\diffp)$ then for any $\bbt \in [0,\ep]^k$ such that $\delta(\bbt,\diffp_g) = \delta_{\opt}(\nonadaptBR^k,\diffp_g)$, we must have 
\[
\diffp_g < \sum_{i=1}^k t_i < \diffp_g + k\diffp
\]

\end{lemma}

\begin{proof}
Using the fact that $q_{t} = e^t p_t$ and $(1 - q_t) = e^{t-\diffp} (1 - p_t)$, we equivalently have 

\[
\delta(\bbt,\diffp_g)  = \sum_{S \subseteq \{1,...,k\}} \prod_{i \notin S}p_{t_i} \prod_{i \in S} (1 - p_{t_i}) \max\left\{e^{\sum t_i - |S|\diffp} - e^{\diffp_g}, 0 \right\}
\]

It then follows that if $\sum t_i \leq \diffp_g$ we must have
\[
\max\left\{e^{\sum t_i - |S|\diffp} - e^{\diffp_g}, 0 \right\} = 0
\]
for any $S$ and so $\delta(\bbt,\diffp_g) = 0$. However, if $\diffp_g < k\diffp$, then there must exist $\bbt$ such that $t_i < \diffp$ for each $i$ and $\sum t_i > \diffp_g$. Setting $S = \emptyset$ we must have $p_{t_i} > 0$ for all $i$ and $\max\{ e^{\sum t_i} - e^{\diffp_g} , 0\} > 0$. Therefore, $\delta_{\opt}(\nonadaptBR^k,\diffp_g) > 0$ and if $\sum t_i \leq \diffp_g$ we must have $\delta(\bbt,\diffp_g) < \delta_{\opt}(\nonadaptBR^k,\diffp_g)$.

Similarly, if $\sum t_i \geq \diffp_g + k\diffp$ we must have the following for any subset $S$
\[
\max\left\{e^{\sum t_i - |S|\diffp} - e^{\diffp_g}, 0 \right\} = e^{\sum t_i - |S|\diffp} - e^{\diffp_g}
\]
We then have the following,
\begin{multline*}
\delta(\bbt,\diffp_g) = \sum_{S \subseteq \{1,...,k\}} \prod_{i \notin S}p_{t_i} \prod_{i \in S} (1 - p_{t_i}) \left(e^{\sum t_i - |S|\diffp} - e^{\diffp_g}\right) 
\\
= \sum_{S \subseteq \{1,...,k\}} \prod_{i \notin S} q_{t_i} \prod_{i \in S} (1 - q_{t_i}) - e^{\diffp_g}\prod_{i \notin S}p_{t_i} \prod_{i \in S} (1 - p_{t_i}) = 1 - e^{\diffp_g}
\end{multline*}

By the same reasoning, we have $\delta(\bbt,\diffp_g) > 1 - e^{\diffp_g}$ if $e^{\sum t_i - |S|\diffp} - e^{\diffp_g} < 0$ for some $S \subseteq \{1,\cdots,k\}$ and all $t_i \in (0,\diffp)$, which implies $p_{t_i} \in (0,1)$ for all $i$. Accordingly, we have $\delta(\bbt,\diffp_g) > 1 - e^{\diffp_g}$ if $\sum t_i < \diffp_g + k \diffp$, and if $\diffp_g > -k\diffp$, there must exist positive $t_i$ such that $\sum t_i < \diffp_g + k \diffp$. Therefore  if $\sum t_i \geq \diffp_g + k\diffp$, we must have $\delta(\bbt,\diffp_g) < \delta_{\opt}(\nonadaptBR^k,\diffp_g)$.
\end{proof}

The next lemma shows that taking the average of some $t_i,t_j$ can only increase the value of $\delta(\bbt,\diffp_g)$. Further, this will strictly increase the $\delta$ when the $t_i$ satisfy the condition of the lemma above. We will be able to easily conclude from this that $\delta$ cannot be optimal if $t_i \neq t_j$ for some $i,j$

\begin{lemma}\label{lem:gen_convexity} For any $\diffp_g \in \R$ and $\bbt \in [0,\ep]^k$,
\[
\delta(\bbt,\diffp_g) \leq \delta\left( \left(\frac{t_1 + t_2}{2}, \frac{t_1 + t_2}{2}, t_3,...,t_k\right), \diffp_g\right)
\]
Further, the inequality is strict whenever $\diffp_g< \sum t_i <\diffp_g + k\diffp $ and $t_1 \neq t_2$.
\end{lemma}

The proof of this lemma will require quite a bit of technical detail which we relegate to Appendix~\ref{sec:nonadaptiveapp}. 
We then have the immediate corollary.

\begin{corollary}\label{cor:all_t_equal}
For any $\diffp_g \in (-k\diffp,k\diffp)$ we must have the following for any $\bbt \in [0,\ep]^k$ such that there exists some $t_i \neq t_j$
\[
\delta(\bbt,\diffp_g) < \delta_{\opt}(\nonadaptBR^k,\diffp_g).
\]
\end{corollary}

\begin{proof}
We will prove by contradiction and suppose $\delta(\bbt,\diffp_g) = \delta_{\opt}(\nonadaptBR^k,\diffp_g) $ and $t_i \neq t_j$ for some pair of indices. Note that $\delta(\bbt,\diffp_g)$ is equal under permutation of the indices in $\bbt$, so without loss of generality, we let $t_1 \neq t_2$. From Lemma~\ref{lem:bound_sum_t}, we must have $\diffp_g < \sum t_i < \diffp_g+ k\diffp$. We then apply Lemma~\ref{lem:gen_convexity} to get our contradiction
\[
\delta(\bbt,\diffp_g) < \delta\left(\frac{t_1 + t_2}{2}, \frac{t_1 + t_2}{2}, t_3,...,t_k\right) \leq \delta_{\opt}(\cM_{\text{BR}}^k,\diffp_g) 
\]
\end{proof}

We now prove the simplified formula for the optimal privacy parameters for the family $\nonadaptBR^k$ of $\diffp$-BR mechanisms, although in the next subsection, we show that we can restrict the range $[0,\diffp]$ that the $\sup$ is over  a smaller set.

\begin{lemma}\label{lem:bin_reduction_delta}
For any $\diffp_g \in \R$ and $\diffp \geq 0$
\begin{equation}
\delta_{\opt}(\nonadaptBR^k,\diffp_g) = \sup_{t \in [0,\diffp]} \sum_{i = 0}^k {k \choose i} p_{t}^{k-i}(1 - p_{t})^{i} \max\left\{\left( e^{kt - i\diffp} - e^{\diffp_g} \right), 0 \right\} 
\label{eq:opt_sup_t}
\end{equation}
\end{lemma}

\begin{proof}
By Lemma~\ref{lem:non_interactive_GRR} and our definition for $\delta(\bbt,\diffp_g)$ given in \eqref{eq:delta_t},
$
\delta_{\opt}(\nonadaptBR^k,\diffp_g) = \sup_{\bbt \in [0,\ep]^k} \delta(\bbt,\diffp_g).
$
From Corollary~\ref{cor:all_t_equal} we know that for $\diffp_g \in (-k\diffp,k\diffp)$,
\[
\delta_{\opt}(\nonadaptBR^k,\diffp_g) = \sup_{t \in [0,\diffp]} \delta(t,\ldots,t,\diffp_g).
\]
Furthermore, we know if $\diffp_g \geq k\diffp$ then $\delta(\bbt,\diffp_g) = 0$ for any $\bbt \in [0,\ep]^k$, and also if $\diffp_g \leq -k\diffp$ then $\delta(\bbt,\diffp_g) = 1 - e^{\diffp_g}$ for any $\bbt \in [0,\ep]^k$. Therefore,
\begin{multline*}
\delta_{\opt}(\nonadaptBR^k,\diffp_g) = \sup_{t \in [0,\diffp]} \sum_{S \subseteq \{1,...,k\}} \prod_{i \notin S}p_{t} \prod_{i \in S} (1 - p_{t}) \max\left\{e^{kt - |S|\diffp} - e^{\diffp_g}, 0 \right\} 
\\
= \sup_{t \in [0,\diffp]} \sum_{S \subseteq \{1,...,k\}} p_{t}^{k - |S|}  (1 - p_{t})^{|S|} \max\left\{e^{kt - |S|\diffp} - e^{\diffp_g}, 0 \right\}
\end{multline*}
For each $i \in \{0, 1, \cdots, k\}$ there are ${k \choose i}$ subsets $S \subseteq \{1,\ldots,k\}$ such that $|S| = i$, and grouping these together gives our desired equality.
\end{proof}

\subsection{Efficiently computing the optimal composition bound}

Now that we have a much simpler formulation of the optimal composition for BR mechanisms in \eqref{eq:opt_sup_t}, we will solve for the $t \in [0,\diffp]$ that maximizes this expression. Ultimately, we will show that there are only $k$ different candidate values of $t$ that maximizes $\delta((t,t, \cdots, t), \diffp_g)$, and give explicit expressions for these candidate values of $t$. These explicit expressions will also be necessary in later sections when we show that there is a difference between the adaptive and nonadaptive setting.

Since we no longer need to consider any $\bbt \in [0,\ep]^k$ where $\bbt$ is not a scalar times the all ones vector, we will simplify our notation to be
\begin{equation}
\delta^k(t,\ep_g) \defeq 
\sum_{i = 0}^k {k \choose i} p_{t}^{k-i}(1 - p_{t})^{i} \max\left\{\left( e^{kt - i\diffp} - e^{\diffp_g} \right), 0 \right\}.
\label{eq:delta_k}
\end{equation}

Given that we want to find the $t$ which maximizes this expression, our goal will be to take the partial derivative of this function with respect to $t$. The maximization within the expression will make this more difficult, however, because the maximization is over a variable term and zero, we will always be able to write $\delta_\opt$ in terms of the following function $F_\ell$ for some $\ell \in \{0, \cdots, k\}$ that will depend on $t$.

\begin{equation}
F_{\ell}(t,\diffp_g) \defeq \sum_{i = 0}^{\ell} {k \choose i} p_{t}^{k-i}(1 - p_{t})^{i} \left( e^{kt - i\diffp} - e^{\diffp_g} \right).
\label{eq:F_t}
\end{equation}
This function is differentiable and we show its relation to $\delta^k(t,\ep_g)$.

\begin{lemma}\label{lem:f_t_equal_delta}
For any $\ep_g \in \R$, $\diffp \geq 0$, and $t \in [0,\ep]$, there must exist some $\ell \in [k]$ such that 

\[
\delta^k(t,\ep_g) = F_{\ell}(t,\ep_g).
\]

\end{lemma}

\begin{proof}
Note that $e^{kt - i\diffp} - e^{\diffp_g}$ decreases as $i$ increases, which implies that for any $t \in [0,\diffp]$ there must exist some $\ell$ such that $\max\{e^{kt - i\diffp} - e^{\diffp_g}, 0\} = e^{kt - i\diffp} - e^{\diffp_g}$ for all $i \leq \ell$ and $\max\{e^{kt - i\diffp} - e^{\diffp_g}, 0\} = 0$ for all $i > \ell$. Therefore, because $p_t$ and $(1 - p_t)$ are non-negative we have

\[
\delta^k(t,\diffp_g) = F_{\ell}(t,\diffp_g).
\]

\end{proof}

It then follows that optimizing over $t \in [0,\ep]$ for $\delta^k(t,\ep_g)$ can be reduced to optimizing over $t \in [0,\ep]$ for each $F_{\ell}(t,\ep_g)$.

\begin{corollary}\label{cor:f_t_max_over_all}
	For any $\diffp_g \in \R$ and $\diffp \geq 0$,
	
	\[
	\delta_{\opt}(\nonadaptBR^k,\diffp_g)
	= \max_{0 \leq \ell \leq k} \{ \sup_{t \in [0,\diffp]} F_{\ell}(t,\diffp_g) \}.
	\]

\end{corollary}

\begin{proof}
Follows immediately from Lemma~\ref{lem:f_t_equal_delta} and because for any $\ep_g$ and $t\in [0,\ep]$, by definition $F_{\ell}(t,\ep_g) \geq \delta^k(t,\ep_g)$ for all $\ell$.

\end{proof}

We will now individually solve each $\sup_{t \in [0,\diffp]} F_{\ell}(t,\diffp_g)$, which does not contain a maximization term and is differentiable. Our ultimate goal will be to solve $\frac{\partial F_{\ell}(t,\diffp_g)}{\partial t} = 0$, and we want explicit expressions for $t$, which will require a simple formulation of the partial derivate with respect to $t$. These explicit expressions will also be necessary for proving that there is a gap between the nonadaptive and adaptive settings. 
The proof for this will become quite involved with some surprisingly nice cancellation, and we relegate the details to Appendix~\ref{sec:nonadaptiveapp}.

\begin{lemma}\label{lem:partial_derivative_full}
For $\diffp_g \in \R$, $\diffp \geq 0$, and $0 \leq \ell \leq k$
\[
\frac{\partial F_{\ell}(t,\diffp_g)}{\partial t} = (k - \ell) {k \choose \ell} p_{t}^{k - 1 - \ell} (1 - p_{t})^{\ell} \frac{1}{1 - e^{-\diffp}} \left(e^{\diffp_g - t} - e^{kt - (\ell + 1)\diffp} \right).
\]

\end{lemma}

In order to prove that there is a gap between composition of adaptive and nonadaptive BR mechanisms, we will further utilize this exact characterization of the partial derivative to give a strict interpretation of the set of $t$ that can achieve a maximization of our full expression. However, for giving an efficiently computable expression for optimal composition, the following simple corollary will suffice.

\begin{corollary}\label{cor:vals_of_t}
For $\diffp_g \in \R$, $\diffp \geq 0$, and $0 \leq \ell \leq k$ 
\[
\arg\sup_{t \in [0,\diffp]}  F_{\ell}(t,\diffp_g)  \in \left\{0,\diffp, \frac{\diffp_g + (\ell + 1) \diffp}{k + 1} \right\}.
\]
\end{corollary}

\begin{proof}
Note that $p_t = 1$ when $t = 0$ and $p_t = 0$ when $t = \diffp$. Therefore $\frac{\partial F_{\ell}(t,\diffp_g)}{\partial t} = 0$ when $t \in \{0,\diffp\}$ or when $\diffp_g - t = kt - (\ell + 1)\diffp$ which evaluates to $t = \frac{\diffp_g + (\ell + 1)\diffp}{k + 1}$.
\end{proof}

We can now prove our main theorem for this section that gives an efficient computation of optimal composition in the non-adaptive setting, which we restate here.

\nonadaptive*

\begin{proof}
From Lemma~\ref{lem:bin_reduction_delta} we have

\[
\delta_{\opt}(\nonadaptBR^k,\ep_g) = \sup_{t \in [0,\ep]} \delta^k(t,\ep_g)
\]

From Lemma~\ref{lem:f_t_equal_delta} and Corollary~\ref{cor:f_t_max_over_all} we can restrict our consideration to values of $t \in [0,\ep]$ that maximize $F_{\ell}(t,\ep_g)$ for some $\ell \in [k]$. Applying Corollary~\ref{cor:vals_of_t} we can then restrict our consideration to $t_{\ell}$ for all $\ell \in [k]$, along with $0$ and $\ep$.
Note that $p_t = 1$ when $t = 0$ and $p_t = 0$ when $t = \ep$, so it is straightforward to verify that $\delta^k(0,\ep_g) = \delta^k(\ep,\ep_g) = \max\{1-e^{\ep_g},0\}$ for any $\ep_g$. In the proof of Lemma~\ref{lem:bound_sum_t}, we showed that $\delta_{\opt}(\nonadaptBR^k,\ep_g) > 0$ and $\delta_{\opt}(\nonadaptBR^k,\ep_g) > 1 - e^{\ep_g}$ when $\ep_g \in (-k\ep,k\ep)$, so it is irrelevant whether we include $0,\ep$ in this setting. Finally, if $\ep_g \notin (-k\ep,k\ep)$, then from Lemma~\ref{lem:ep_g_edge_case} we have $\delta^k(t,\ep_g) = \delta_{\opt}(\mbr^k,\ep_g)  = \max\{1 - e^{\ep_g},0\}$ for any $t$.

For the running time, first note that for any $t$ we can compute $p_t^k(e^{kt}-e^{\ep_g})$ in $O(k)$ time. Further, for any $t$, if we are given the values ${k \choose i} p_t^{k-i} (1 - p_t)^i$ and $e^{kt - i\ep} $, then we can compute ${k \choose i + 1} p_t^{k-(i + 1)} (1 - p_t)^{i+1}$ and $e^{kt - (i + 1)\ep} $ in $O(1)$ time. Our running time of $O(k^2)$ then immediately follows.
\end{proof}

%% file: interactive-gap.tex
\section{Adaptive optimal composition}\label{sec:gap} 

In this section, we give the formulation for the optimal composition of BR mechanisms that can be chosen adaptively, which will be recursively defined and intractable even for reasonable $k$. We see no way to simplify this formulation and believe that exact computation (or even approximate) is hard, but we leave that for future work.
We further show that there is in fact a gap between the optimal composition bound in the adaptive and the nonadaptive cases for all $k \geq 2$, and that this gap exists for almost all non-trivial $\diffp_g$.

We will set up some notation that is similar to what we presented in Section~\ref{sect:added_power}, although we extend it here to the heterogeneous case, where $\diffp_1, \cdots, \diffp_k$ need not be the same. Given some fixed $\diffp_1,\ldots,\diffp_k$, and mechanisms $(\cM_1,\ldots,\cM_k)$ be such that $\cM_i$ is the class of $\diffp_i$-BR mechanisms. We then define the following family of mechanisms, which generalizes the homogeneous case $\abr^k$ given in \eqref{eq:class_adaptive_comp},

\begin{equation}
\abr^{1:k} \defeq \{\AdaComp(\cA,(\cM_1, \cdots, \cM_k),\cdot): \text{ adversary } \cA\}.
\label{eq:class_gen_adaptive_comp}
\end{equation}

The formulations and proofs in this section will rely upon recursive definitions, and it then becomes necessary to define the adaptive composition for different families of mechanisms, i.e. $\abr^{\ell:k} \defeq \{\AdaComp(\cA, (\cM_\ell,\ldots,\cM_k),\cdot): \text{ adversary } \cA\}$ for $\ell \in [k]$.

These definitions will then allow us to give an explicit recursive formulation of the optimal composition bounds for the $k$-fold adaptive composition of BR mechanisms. This formulation will follow from Corollary~\ref{cor:reduc} which allows us to restrict our consideration to deterministically choosing $t_i$ for our generalized random response, where this choice is conditional upon the previous outcomes. The proof will be straightforward, but notationally heavy.

\begin{lemma} \label{lem:adap_recursion_hetero}
Let $\abr^{1:k}$ be the class of adaptive $k$-fold composition of $\diffp_i$-BR mechanisms given in \eqref{eq:class_gen_adaptive_comp}, then for any $\diffp_g \in \R$ and setting
$
\delta_{\opt}(\abr^{k+1:k},\diffp_g) = \max\{1 - e^{\diffp_g}, 0 \}
$
we have,

\[
\delta_{\opt}(\abr^{1:k},\diffp_g) = \sup_{t_1 \in [0,\diffp_1]} \left\{ q_{\diffp_1,t_1} \delta_{\opt}(\abr^{2:k}, \diffp_g - t_1) + (1 - q_{\diffp_1,t_1})  \delta_{\opt}(\abr^{2:k}, \diffp_g + \diffp_1 - t_1) \right\}
\]
\end{lemma}

\begin{proof}

We will prove this claim by induction, where the key will be to apply Corollary~\ref{cor:reduc} which gives that we can equivalently restrict our consideration to adversaries without their own randomness and only consider mechanisms in the generalized randomized response class.

For our base case of $k = 1$, we have from Fact~\ref{fact:delta_opt}, Corollary~\ref{cor:reduc}, and using $\abr^{1:1}$ with privacy parameter $\diffp_1$,
\[
\delta_{\opt}(\abr^{1:1},\diffp_g) = \sup_{t_1 \in [0,\diffp_1]} \sup_{b_1 \in \{0,1\}} \sum_{y_1 \in \{0,1\}} \max \left\{ \Pr[\grr{\diffp_1,t_1}(b_1) = y_1] - e^{\diffp_g}\Pr[\grr{\diffp_1,t_1}(1 - b_1) = y_1], 0 \right\}.
\]

The symmetry of generalized random response from Claim~\ref{claim:symmetric} implies that we can fix $b_1 = 0$, and this reduces to 
\[
\delta_{\opt}(\abr^{1:1},\diffp_g) = \sup_{t_1 \in [0,\diffp_1]} \big\{ \max \left\{ q_{\diffp_1,t_1} - e^{\diffp_g}p_{\diffp_1,t_1}, 0 \right\} + \max \left\{ (1 - q_{\diffp_1,t_1}) - e^{\diffp_g}(1 - p_{\diffp_1,t_1}), 0 \right\} \big\}.
\]

Using the fact that $q_{\diffp_1,t_1} = e^{t_1}p_{\diffp_1,t_1}$ and $(1 - q_{\diffp_1,t_1}) = e^{t_1 - \diffp_1}(1 - p_{\diffp_1,t_1})$, this reduces to our desired equality. We then assume for $k-1$, and again applying Fact~\ref{fact:delta_opt} and Corollary~\ref{cor:reduc} we have the following  for the deterministic adversary $\cA = (\cDe,\emptyset)$ without its own source of randomness and letting $\vcRR = (\cRR_1, \ldots,\cRR_k)$ be the class such that $\cRR_i \defeq \{\grr{\diffp_i,t_i} : t_i \in [0,\diffp_i]\}$, 
\begin{multline*}
\delta_{\opt}(\abr^{1:k},\diffp_g) = 
\\
\sup_{ \cA = (\cDe, \emptyset)} 
\sum_{\bby \in \{0,1\}^k} \max\left\{ \Pr[\AdaComp(\cA,\vcRR,b) = \bby] - e^{\diffp_g}\Pr[\AdaComp(\cA,\vcRR,1-b) = \bby],0\right\}.
\end{multline*}
We will expand this term by considering the first round where some $t_1 \in [0,\diffp_1]$ is chosen deterministically. Once again, we use the symmetry of generalized random response from Claim~\ref{claim:symmetric} to simply set $b_1 = 0$. The next choices are then dependent on this outcome, so the full expression becomes 
\begin{multline*}
\delta_{\opt}(\abr^{1:k},\diffp_g) = 
\\
\sup_{t_1 \in [0,\diffp_1]} \sum_{y_1 \in \{0,1\}}
\sup_{ \cA = (\cDe,\emptyset)} \bigg\{
\sum_{\bby \in \{0,1\}^{k-1}} 
\max\big\{\Pr[\grr{\diffp_1,t_1}(0) = y_1] \Pr[\AdaComp(\cA,(\cRR_{2},\cdots, \cRR_k),b) = \bby] 
\\
- e^{\diffp_g}\Pr[\grr{\diffp_1,t_1}(1) = y_1]\Pr[\AdaComp(\cA,(\cRR_{2},\cdots, \cRR_k),1-b) = \bby],0\big\} \bigg\}.
\end{multline*}

Again, we use the fact that $q_{\diffp_1,t_1} = e^{t_1}p_{\diffp_1,t_1}$ and $(1 - q_{\diffp_1,t_1}) = e^{t_1 - \diffp_1}(1 - p_{\diffp_1,t_1})$ to pull them outside of the maximum in the expression, so that for $y_1 = 0$ the inner term then reduces to  
\begin{multline*}
q_{\diffp_1,t_1} \Bigg( 
\sup_{ \cA=(\cDe,\emptyset)} \bigg\{
\sum_{\bby \in \{0,1\}^{k-1}} 
\max\big\{ \Pr[\AdaComp(\cA,(\cRR_2, \cdots, \cRR_{k}),b) = \bby] 
\\
- e^{\diffp_g - t_1}\Pr[\AdaComp(\cA,(\cRR_2, \cdots, \cRR_{k}),1-b) = \bby],0\big\} \bigg\} \Bigg)
\\
= q_{\diffp_1,t_1} \cdot \delta_{\opt}(\abr^{2:k},\diffp_g - t_1).
\end{multline*}

This similarly follows for $y_1 = 1$, and we have our desired claim.
\end{proof}

Unfortunately, straightforward computation of this formulation is intractable, and we conjecture that it has a similar hardness result as in \citet{MurtaghVa16}, even in the homogenous setting. In later sections, we give improved bounds on adaptive composition for BR mechanisms, but for this section we instead focus on proving that there is indeed a gap between this optimal formulation and our formulation for the nonadaptive setting given in Theorem~\ref{thm:non-adaptive}. Further, we show that this gap exists in the homogenous setting for all $k$ and almost all choices of $\diffp_g$. 

We now state the main result of this section, where we prove each claim in Lemmas~\ref{lem:gapLEM} and \ref{lem:noGap}, respectively.  

\gap*

\subsection{Gap between adaptive and nonadaptive optimal composition}\label{subsec:gap}
In this section we show that there is a gap in the privacy loss between the adaptive and nonadaptive setting for BR mechanisms. Furthermore, we want to prove that this gap exists for all $k\geq 2$ and most $\diffp_g$. In fact, the only values of $\diffp_g$ in which the privacy loss is equivalent is when $\diffp_g$ is almost the bound from basic composition. 

The general idea for proving the gap will be to also give the recursive definition for the nonadaptive optimal composition that must fix $t$ for each recursive call. The goal will then be to show that at some point within this recursion the summation will strictly increase if the value for $t$ is changed. This will require that we first fully characterize the possible values of $t$ for the nonadaptive optimal composition. Fortunately, most of the heavy lifting in this regard was done in the previous section. With this characterization, we show that there is a gap when $k = 2$, and then further show that we can apply this gap for $k \geq 2$.

We will restrict our consideration to the simpler homogenous setting in which $\diffp_i = \diffp$ for all $i$, and use $\abr^k$ as defined in \eqref{eq:class_adaptive_comp} and $\nonadaptBR^k$ is the class of nonadaptive composed $\diffp$-BR mechanisms as in \eqref{eq:nonadaptBR}.  We know that we can instead just restrict our consideration to the class of generalized random response, and the key to our the proof will be that we will be able to specify exactly which values of $t_1,\ldots,t_k$ maximize the privacy loss for the nonadaptive setting. We define this set as in terms of $\delta(\bbt, \diffp_g)$ from \eqref{eq:delta_t},

\[
t_{\opt}(\nonadaptBR^k,\diffp_g) \defeq \{\bbt \in [0,\diffp]^k: \delta(\bbt,\diffp_g) = \delta_{\opt}(\nonadaptBR^k,\diffp_g)\}.
\]

From Corollary~\ref{cor:all_t_equal}, we know that this set cannot contain any $\bbt\in [0,\diffp]^k$ such that $t_i \neq t_j$ in the interesting setting where $\diffp_g \in (-k\diffp,k\diffp)$.
For the remainder of this section, we instead consider the definition to equivalently be

\[
t_{\opt}(\nonadaptBR^k,\diffp_g) = \{t \in [0,\diffp]: \delta^k(t,\diffp_g) = \delta_{\opt}(\nonadaptBR^k,\diffp_g)\}
\]
because when $\diffp_g \notin (-k\diffp,k\diffp)$ then there is not a gap between adaptivity and non-adaptivity, so we ignore this setting. We will further utilize our proofs from the previous section to show that we can further restrict this set.

\begin{lemma}\label{lem:set_of_opt_t}
Let $\diffp \geq 0$.  If $\diffp_g \in (-k\diffp,k\diffp)$, then 
\[
t_{\opt}(\nonadaptBR^k,\diffp_g) \subseteq \left\{ \frac{\diffp_g + (\ell + 1)\diffp}{k + 1} : \ell \in \{0, \cdots, k-1\} \right\}  \cap (0,\diffp).
\]
\end{lemma}

\begin{proof}
From Lemma~\ref{lem:f_t_equal_delta} and Corollary~\ref{cor:f_t_max_over_all} we can restrict our consideration to values of $t \in [0,\diffp]$ that maximize $F_{\ell}(t,\diffp_g)$ for some $\ell \in [k]$. Furthermore, $F_{\ell}(t,\diffp_g)$ can only maximized at the endpoints of the interval or whenever $\frac{\partial F_{\ell}(t,\diffp_g)}{\partial t} = 0$. Thus, from Corollary~\ref{cor:vals_of_t} we have

\[
t_{\opt}(\nonadaptBR^k,\diffp_g) \subseteq \left\{t_{\ell}^* = \frac{\diffp_g + (\ell + 1)\diffp}{k + 1} : \ell \in \{0,k\} \right\}  \cup \{0,\diffp\}.
\]
By definition, we can remove all values outside of $[0,\diffp]$, so it then suffices to show that we can also remove $\{0,\diffp,t_k^*\}$. Note that $p_t = 1$ when $t = 0$ and $p_t = 0$ when $t = \diffp$ and recall $\delta^k(t,\diffp_g)$ from \eqref{eq:delta_k}, so it is straightforward to verify that $\delta^k(0,\diffp_g) = \delta^k(\diffp,\diffp_g) = \max\{1-e^{\diffp_g},0\}$ for any $\diffp_g$. In the proof of Lemma~\ref{lem:bound_sum_t}, we showed that $\delta_{\opt}(\nonadaptBR^k,\diffp_g) > 0$ and $\delta_{\opt}(\nonadaptBR^k,\diffp_g) > 1 - e^{\diffp_g}$ when $\diffp_g \in (-k\diffp,k\diffp)$, which implies $0,\diffp \notin t_{\opt}(\nonadaptBR^k,\diffp_g) $. 

It then suffices to show $t_k^* \notin t_{\opt}(\nonadaptBR^k,\diffp_g) $. If $\diffp_g > 0$, then $t_k^* > \diffp$, so we only need to consider $\diffp_g \leq 0$. Note that $kt_k^* = k(\frac{\diffp_g}{k + 1} + \diffp)$, so for any $i \leq k$ we have $kt_k^* - i\diffp \geq \frac{k}{k+1}\diffp_g$ which implies

\[
\max\left\{e^{kt_k^* - i\diffp} - e^{\diffp_g}, 0 \right\} = e^{kt_k^* - i\diffp} - e^{\diffp_g}.
\]

Therefore, $\delta^k(t_k^*,\diffp_g) = 1 - e^{\diffp_g}$ and from above we know $\delta_{\opt}(\nonadaptBR^k,\diffp_g) > 1 - e^{\diffp_g}$ when $\diffp_g \in (-k\diffp,k\diffp)$, which implies  $t_k^* \notin t_{\opt}(\nonadaptBR^k,\diffp_g) $ as desired.
\end{proof}

We now want to show that we can write the optimal nonadaptive composition in a similar form as the adaptive composition. This recursive formulation will then fix a value $t$ throughout the recursion and $\delta_{\opt}(\nonadaptBR^k,\diffp_g)$ is then just the maximum value of this recursion over all $t \in [0,\diffp]$.

\begin{corollary}\label{cor:non_adap_recursion}
For $k \geq 1$ and for $\delta^k(t,\diffp_g)$ from \eqref{eq:delta_k}, we have  $\delta^0(t,\diffp_g) = \max\{1 - e^{\diffp_g}, 0\}$ and
\[
\delta^k(t,\diffp_g) = q_t\delta^{k-1}(t,\diffp_g - t) + (1 - q_t)\delta^{k-1}(t,\diffp_g + \diffp - t).
\]
\end{corollary}

We relegate the proof of this corollary to Appendix~\ref{sec:gap-proofs}. Now that the formulations are similar, we show the intuitive fact that if at any point in the recursion either it is the case that either 1) switching the value of $t$, or 2) switching to the adaptive setting, will strictly increase that $\delta_{\opt}$ then there must be a gap between the nonadaptive and adaptive setting.

\begin{lemma}\label{lem:recursive_gap}
Fix the individual privacy parameter $\diffp> 0$, some global privacy parameter $\diffp_g\in \R$ and $k \geq 2$, along with some $t \in t_{\opt}(\nonadaptBR^k,\diffp_g)$, if there exist $0 \leq \ell' \leq \ell < k$ such that either $
\delta_{\opt}(\nonadaptBR^{k-\ell},\diffp_g - \ell t + \ell'\diffp) < \delta_{\opt}(\abr^{k-\ell},\diffp_g - \ell t + \ell'\diffp) 
$
or $t \notin t_{\opt}(\nonadaptBR^{k-\ell},\diffp_g - \ell t + \ell'\diffp)$, then we must have 
\[
\delta_{\opt}(\nonadaptBR^k,\diffp_g) < \delta_{\opt}(\abr^{k}, \diffp_g).
\]
\end{lemma}

This lemma will actually require quite a bit of technical detail, so we instead give a proof in Appendix~\ref{sec:gap-proofs}.
With this property and our characterization of $t_{\opt}(\abr^{k},\diffp_g)$, we now show that there is a gap for the base case of $k = 2$.

\begin{lemma}\label{lem:gap_base_case}
For any $\diffp_g \in (-\diffp/2,\diffp/2)$ we have 
\[
\delta_{\opt}(\nonadaptBR^2,\diffp_g) < \delta_{\opt}(\abr^2,\diffp_g) 
\]
\end{lemma}

\begin{proof}
From Lemma~\ref{lem:set_of_opt_t}, we know that there exists and $\ell \in \{ 0,1\}$ such that $t_{\ell} =  \frac{\diffp_g + (\ell + 1)\diffp}{3} \in t_{\opt}(\nonadaptBR^2,\diffp_g)$. Furthermore, if both $\diffp_g - t_{\ell}$ and $\diffp_g - t_{\ell} + \diffp$ are in $(-\diffp,\diffp)$, then we also must have $t_{\opt}(\nonadaptBR^1,\diffp_g - t_{\ell}) = \frac{\diffp_g - t_{\ell} + \diffp}{2}$ and  $t_{\opt}(\nonadaptBR^1,\diffp_g - t_{\ell} + \diffp ) = \frac{\diffp_g - t_{\ell} + 2\diffp}{2}$ which implies $t_{\opt}(\nonadaptBR^1,\diffp_g - t_{\ell}) \neq t_{\opt}(\nonadaptBR^1,\diffp_g - t_{\ell} + \diffp )$.

Therefore, by Lemma~\ref{lem:recursive_gap} it suffices to show that both $\diffp_g - t_{\ell}$ and $\diffp_g - t_{\ell} + \diffp$ are in $(-\diffp,\diffp)$, which is equivalent to showing $\diffp_g - t_{\ell} \in (-\diffp,0)$. Plugging in for $t_{\ell}$ we then have

\[
\diffp_g - \frac{\diffp_g + (\ell + 1)\diffp}{3} \in (-\diffp,0) 
\qquad \Leftrightarrow \qquad
\diffp_g \in \left(\frac{(\ell - 2)\diffp}{2}, \frac{(\ell + 1)\diffp}{2}\right)
\]

which holds for $\ell \in \{0,1\}$ by our assumption that $\diffp_g \in (-\diffp/2,\diffp/2)$.
\end{proof}

We will then apply this base case to the more general case for certain conditions by applying Lemma~\ref{lem:recursive_gap}.

\begin{lemma}\label{lem:use_2_case}
Given some $\diffp_g \in (-k\diffp,k\diffp)$ and $t \in t_{\opt}(\nonadaptBR^k,\diffp_g)$. For $k \geq 4$, if $\diffp_g - (k - 2) t < \diffp/2$ and $\diffp_g - (k-2)t + (k-2)\diffp > - \diffp/2$, then

\[
\delta_{\opt}(\nonadaptBR^k,\diffp_g) < \delta_{\opt}(\abr^{k},\diffp_g).
\]
\end{lemma}

We relegate the proof of this lemma to Appendix~\ref{sec:gap-proofs} and will use this to show our desired result.

\begin{lemma}\label{lem:gapLEM}
For any $\diffp_g \in [-(k-3)\diffp, (k-3)\diffp]$ and $k \geq 4$ we have 
\[
 \delta_{\opt}(\nonadaptBR^k,\diffp_g) < \delta_{\opt}(\abr^{k},\diffp_g) .
\]
\end{lemma}

\begin{proof}
We will prove for $\diffp_g \in [0,(k-3)\diffp]$ and the case of $\diffp_g \in [-(k-3)\diffp,0]$ follows symmetrically. From Lemma~\ref{lem:set_of_opt_t} we know that for any $t \in t_{\opt}(\nonadaptBR^k,\diffp_g)$ we must have $t = \frac{\diffp_g + (\ell + 1)\diffp}{k + 1}$ for some $0 \leq \ell \leq k-1$. The general idea will then be to show that for any $t_{\ell} = \frac{\diffp_g + (\ell + 1)\diffp}{k + 1}$, if $t_{\ell} \in t_{\opt}(\nonadaptBR^k,\diffp_g)$, then $\delta_{\opt}(\abr^{k},\diffp_g) > \delta_{\opt}(\nonadaptBR^k,\diffp_g)$. We will split this into three cases.

\paragraph{Case I: ($\ell \geq 2$)}

For this setting, we want to show that we can apply Lemma~\ref{lem:use_2_case} where we know $\diffp_g + (k-2)(\diffp - t) \geq -\diffp/2$ for any $t$ because we are assuming $\diffp_g \geq 0$. It then suffices to show that $\diffp_g - (k-2)t_{\ell} < \diffp/2$. Plugging in for $t_{\ell}$ we have

\[
\diffp_g - (k-2)t_{\ell} < \diffp/2 \qquad \Leftrightarrow \qquad
6\diffp_g < \left(2(k-2)(\ell + 1) + (k + 1)\right) \diffp.
\]

By assumption, we know $\diffp_g \leq (k-3)\diffp$, so for $\ell \geq 2$, we have

\[
6\diffp_g \leq 6(k-3)\diffp < (7k - 11)\diffp \leq \left(2(k-2)(\ell + 1) + (k + 1)\right) \diffp.
\]
and therefore $\delta_{\opt}(\abr^{k},\diffp_g) > \delta_{\opt}(\nonadaptBR^k,\diffp_g)$ by Lemma~\ref{lem:use_2_case}. 

\paragraph{Case II: ($\ell = 0$)}

For this setting we have $t_0 = \frac{\diffp_g + \diffp}{k + 1}$. By our assumption that $\diffp_g \in [-(k-3)\diffp, (k-3)\diffp]$, we must have $\diffp_g + \diffp - t_0 \in (-(k-1)\diffp,(k-1)\diffp)$. 
From Lemma~\ref{lem:set_of_opt_t} we then know $t_{\opt}(\nonadaptBR^{k-1},\diffp_g + \diffp - t_0) \subseteq \left\{\tfrac{\diffp_g + \diffp - t_0 + (\ell' + 1)\diffp}{k}: \ell' \in \{ 0,k-2\}\right\}$. 
We further see that for any $\ell' \geq 0$,
\[
 \frac{\diffp_g + \diffp}{k + 1} < \frac{\diffp_g + \diffp}{k} \leq \frac{\diffp_g + \diffp - t_0 + (\ell' + 1)\diffp}{k}.
\]
This implies $t_0 \notin t_{\opt}(\nonadaptBR^{k-1},\diffp_g + \diffp - t_0)$ and so $\delta_{\opt}(\abr^{k},\diffp_g) > \delta_{\opt}(\nonadaptBR^k,\diffp_g)$ by Lemma~\ref{lem:recursive_gap}.

\paragraph{Case III: ($\ell = 1$)}

This will follow from the same argument as the previous case. For this setting we have $t_1 = \frac{\diffp_g + 2\diffp}{k+1}$. Once again, we use our more restrictive assumption that $\diffp_g \in [0,(k-3)\diffp]$, and therefore $\diffp_g + 2(\diffp - t_1) \geq -(k-2)\diffp$. Furthermore, we have

\[
\diffp_g + 2\left(\diffp - \frac{\diffp_g + 2\diffp}{k  + 1}\right) = \frac{k-1}{k+1}\left(\diffp_g + 2\diffp\right) \leq \frac{(k-1)^2}{k+1}\diffp < (k-2)\diffp
\]
where the last step follows because $(k-1)^2 < (k + 1)(k-2)$ for $k > 1$.
Thus $\diffp_g + 2(\diffp - t_1) \in (-(k-2)\diffp,(k-2)\diffp)$ and by Lemma~\ref{lem:set_of_opt_t}, $t_{\opt}(\nonadaptBR^{k-2},\diffp_g + 2(\diffp - t_1)) \subseteq \left\{\tfrac{\diffp_g + 2(\diffp - t_1) + (\ell'' + 1)\diffp}{k-1} : \ell'' \in \{ 0,k-3\} \right\}$. It then follows that 

\[
\frac{\diffp_g + 2\diffp}{k + 1} = \frac{\diffp_g + 2(\diffp - t_1)}{k-1} < \frac{\diffp_g + 2(\diffp - t_1) + (\ell'' + 1)\diffp}{k-1} 
\]
for any $\ell'' \geq 0$. This implies $t_1 \notin t_{\opt}(\nonadaptBR^{k-2},\diffp_g + 2(\diffp - t_1))$ and so $\delta_{\opt}(\abr^{k},\diffp_g) > \delta_{\opt}(\nonadaptBR^k,\diffp_g)$ by Lemma~\ref{lem:recursive_gap}.

\end{proof}

\subsection{Settings for equivalent adaptive and nonadaptive optimal composition}\label{subsec:nogap}

We also want to show that there is not a gap between adaptive and nonadaptive composition even in the non-trivial setting. More specifically, we will show that there is a gap not only when $\diffp_g \notin (-k\diffp,k\diffp)$ and basic composition can be applied.
We first show that there is no gap for the trivial setting and then will extend this a bit.

\begin{lemma}\label{lem:equal_in_trivial}
For any $\diffp > 0$ and $\diffp_g \geq k\diffp$, we have $\delta_{\opt}(\abr^{k},\diffp_g) = 0$, and for any $\diffp_g \leq -k\diffp$, we have $\delta_{\opt}(\abr^{k},\diffp_g) = 1 - e^{\diffp_g}$.

\end{lemma}

We leave the proof of the trivial setting to Appendix~\ref{sec:gap-proofs}. The basic idea for extending this interval a bit further will simply be to consider the case in which only one outcome can produce a positive probability, or equivalently, all but one outcome can produce a positive probability.

\begin{lemma}\label{lem:upper_no_gap}

For any $\diffp > 0$ and $\diffp_g \geq (k-1)\diffp$ for $k \geq 1$, we have

\[
\delta_{\opt}(\abr^{k},\diffp_g) = \sup_{\bbt \in [0,\diffp]^k} \left\{ \left( \prod_{i=1}^k q_{t_i} \right) \max\{1 - e^{\diffp_g - \sum t_i}, 0\}  \right\}
\]

\end{lemma}

\begin{proof}
We show this inductively. For $k = 1$, if $\diffp_g \geq 0$, then for any $t \in [0,\diffp]$, we must have $\diffp_g + \diffp - t \geq 0$ and $\max\{1 - e^{\diffp_g + \diffp - t},0\} = 0$. This then implies

\[
\delta_{\opt}(\abr^1,\diffp_g) =   \sup_{t \in [0,\diffp]}  q_{t} \max\{1-e^{\diffp_g - t},0\}.
\]

The inductive step for $k \geq 2$ follows equivalently, where for any $t \in [0,\diffp]$, we must have $\diffp_g + \diffp - t \geq (k-1)\diffp$, so from Lemma~\ref{lem:equal_in_trivial}, we have 

\[
\delta_{\opt}(\abr^{k},\diffp_g) = \sup_{t \in [0,\diffp]} q_{t} \delta_{\opt}(\abr^{k-1}, \diffp_g - t) 
\]
and we can then apply our inductive hypothesis because $k - 1\geq 1$ and $\diffp_g - t \geq (k-2)\diffp$, which then gives the desired claim.
\end{proof}

For completeness, we also consider the symmetric case where $\diffp_g$ can be negative, but leave the proof to Appendix~\ref{sec:gap-proofs}.

\begin{lemma}\label{lem:lower_no_gap}
For any $\diffp > 0$ and $\diffp_g \leq -(k-1)\diffp$ with $k \geq 1$, we have
\[
\delta_{\opt}(\abr^{k},\diffp_g) = 1 - e^{\diffp_g} + \sup_{\bbt \in [0,\diffp]^k}\left\{ \left( \prod_{i=1}^k \left(1 - q_{t_i}\right) \right) \left(e^{\diffp_g + k\diffp - \sum t_i } - 1\right) \right\}.
\]
\end{lemma}

We then have the following result that together with Lemma~\ref{lem:gapLEM} covers almost all choices of $\diffp_g \in \R$.

\begin{lemma}\label{lem:noGap}
For any $\diffp_g \geq (k-1)\diffp$ or $\diffp_g \leq -(k-1)\diffp$ we have 

\[
\delta_{\opt}(\abr^{k},\diffp_g) = \delta_{\opt}(\nonadaptBR^k,\diffp_g)
\]
\end{lemma}

\begin{proof}
Each case can be proven directly following similar reasoning as in Lemmas~\ref{lem:upper_no_gap} and ~\ref{lem:lower_no_gap}. But, we can more easily point out that in both cases, Lemmas~\ref{lem:upper_no_gap} and ~\ref{lem:lower_no_gap} imply that the choices of $t_1,\ldots, t_k$ are not adaptively made, so we must have $\delta_{\opt}(\abr^{k},\diffp_g) = \delta_{\opt}(\nonadaptBR^k,\diffp_g)$.
\end{proof}

%% file: mgf.tex
\section{Improved and efficient adaptive composition bounds} 
\label{sec:mgf_approach}
Although we have presented the optimal composition bound in \eqref{lem:adap_recursion_hetero}, directly computing it is intractable.  We then aim to bound the privacy loss with computationally efficient bounds that improve on previous work.  
Given our formulation of the privacy loss in terms of a summation of individual generalized randomized response privacy loss variables, we then follow a similar analysis to concentration inequalities, e.g. Azuma-Hoeffding bounds, by bounding the moments of the privacy loss.  We now present the main result of this section.

\mgfnew*

By \Cref{cor:reduc} we can assume all the component mechanisms are generalized randomized response. 
Let $L_i$ be log likelihood ratio in the $i$-th term in the total privacy loss and we will write $t_i = t_i(y_{1}, \cdots, y_{i-1})$
\begin{align*}
	L_i(y_1,\ldots, y_i)
	&= \log \frac{\Pr[\grr{\diffp_i, t_i}(1)= y_i|y_{1}, \cdots, y_{i-1} ]}{\Pr[\grr{\diffp_i, t_i}(0)= y_i|y_{1}, \cdots, y_{i-1}]}\\
	&=\left\{
		\begin{array}{ll}
		t_i, 		& \text{if } y_i = 1, \\
		t_i-\diffp_i, 		& \text{if } y_i = 0.
		\end{array}
		\right.
\end{align*}
Recall from Corollary~\ref{cor:reduc} that we need only consider deterministic adversaries $\cA = (\emptyset,\cDe)$ and the class of mechanisms $\vcM$ to be the class of generalized randomized response mechanisms $\vcRR$.  For simplicity, let $P$ be the output distribution for $\AdaComp(\cA,\vcM,b)$ and $Q$ be the output distribution for $\AdaComp(\cA,\vcM,1-b)$ on $
\{0,1\}^k$
and $L=\sum_{i=1}^k L_i = \log \frac{Q}{P}$ be the log likelihood ratio of the composed mechanism. Then we have
$$\delta_{\opt}(\abr^{1:k},\diffp_g)= Q[L>\diffp_g]-e^{\diffp_g} P[L>\diffp_g].$$
The classical method introduced in \citet{DworkRoVa10} make two approximations: (1) ignore the negative term $-e^{\diffp_g} P[L>\diffp_g]$ and (2) use moment generating function of $L$ to bound the tail probability $Q[L>\diffp_g]$. We follow (1) and improve on (2) with the help of the reduction in \Cref{cor:reduc}. We present the proof of \Cref{thm:MGFthm} to point out the stages in the analysis where other approaches used weaker bounds.
The initial steps remain consistent:
	\begin{align*}
	\delta_{\opt}(\abr^{1:k},\diffp_g)
	&= Q[L>\diffp_g]-e^{\diffp_g} P[L>\diffp_g]\\
		&\leqslant Q[L>\diffp_g]\\
		&= \Pr\big[\sum L_i> \diffp_g\big]\\
		&\leqslant \inf_{\lambda>0}\Pr\big[e^{\lambda\sum L_i}> e^{\lambda\diffp_g}\big]\\
		&\leqslant \inf_{\lambda>0}e^{-\lambda\diffp_g}\cdot\E\big[e^{\lambda\sum L_i}\big].
	\end{align*}
With a standard conditional probability argument we have the following result.
	\begin{lemma}\label{lem:origin}
		If there is a function $U_i:(0,+\infty)\to \R$ such that for each $i=1,2,\ldots,k$ the following holds for any arbitrary outcomes $y_1,\ldots,y_{i-1}$ of the previous generalized randomized response mechanisms,
		\begin{align*}
			\E_Q[e^{\lambda L_i}\mid y_1,\ldots,y_{i-1}]\leqslant e^{U_i(\lambda)},
		\end{align*}
		then the following holds for any $\lambda>0$,
		\[
		\delta_{\opt}(\abr^{1:k},\diffp_g)
		\leqslant e^{-(\lambda\diffp_g-\sum_{i} U_i(\lambda))}.
		\]
	\end{lemma}
	Different bounds correspond to different choices of $U_i(\lambda)$ in \Cref{lem:origin}, which result in different bounds on $\delta_{\opt}(\abr^{1:k},\diffp_g)$. For example, both \citet{DworkRoVa10} and \citet{DurfeeRo19} utilize the following lemma:
	\begin{lemma}[Hoeffding's lemma]
		If a random variable $X\in[a,b]$ then $\log \E[e^{\lambda X}]\leqslant \frac{1}{8}(b-a)^2\lambda^2 + \lambda \E X$.
	\end{lemma}
	We now walk through the following comparisons with previous work to highlight our improvement.
	\citet{DworkRoVa10} only uses the fact that $L_i\in[-\diffp_i,\diffp_i]$
		(which is weaker than $\diffp$-BR).
		It implies
		\begin{enumerate}[(a)]
			\item $\log \E_Q[e^{\lambda L_i}\mid y_1,\ldots,y_{i-1}]\leqslant \frac{1}{2}\diffp_i^2\lambda^2+\lambda\E_Q[{L_i}\mid y_1,\ldots,y_{i-1}]$
			\item $\E_Q[{L_i}\mid y_1,\ldots,y_{i-1}]\leqslant \diffp_i\tanh\tfrac{\diffp_i}{2}\leqslant\tfrac{1}{2}\diffp_i^2$.
		\end{enumerate}
		For part (b), \citet{DworkRoVa10} used a much rougher estimate. The $\tfrac{1}{2}\diffp_i^2$ upper bound appears in \cite{BunSt16}. For the most refined bound in terms of hyperbolic tangent function, readers can refer to Lemma D.8 in \cite{DongRoSu19}. 

		Combining both (a) and (b), we have $U_i(\lambda) = \frac{1}{2}\diffp_i^2(\lambda^2+\lambda)$, which we refer to as ``Improved DRV10'' in \Cref{fig:U}.
		
		Using the bounded range property from \citet{DurfeeRo19}, we know that for $\ep$-BR there is a $t_i\in[0,\diffp_i]$ such that $a=t_i-\diffp_i,b=t_i$ in Hoeffding's lemma. A similar argument yields $U_i(\lambda) = \frac{1}{2}\diffp_i^2(\frac{1}{4}\lambda^2+\lambda)$, which we label as ``DR19" in \Cref{fig:U}.

		A straightforward improvement could come from a finer treatment of (b). By definition of $L_i$, 
		\begin{align*}
			\E_Q[{L_i}\mid y_1,\ldots,y_{i-1}]
			&= \mathrm{KL}\big(\mathrm{Bern}(q_{\ep_i,t_i})\|\mathrm{Bern}(p_{\diffp_i,t_i})\big)\\
			&= q_{\diffp_i,t_i}\cdot \log\frac{q_{\diffp_i,t_i}}{p_{\diffp_i,t_i}}+(1-q_{\diffp_i,t_i})\cdot \log\frac{1-q_{\diffp_i,t_i}}{1-p_{\diffp_i,t_i}}\\
			&= t_iq_{\diffp_i,t_i}+(t_i-\diffp_i)(1-q_{\diffp_i,t_i})\\
			&= t_i-\frac{\diffp_i}{e^{\diffp_i}-1}(e^{t_i}-1).
		\end{align*}
		A bit of calculus shows the above expression is maximized at $t_i = \log \frac{e^{\diffp_i}-1}{\diffp_i}$, and the value is 
		\[
		\mathrm{maxkl}(\diffp):=\frac{\diffp}{e^\diffp-1}-1-\log\frac{\diffp}{e^\diffp-1}.
		\]
		That is, we have replaced (b) with
		\begin{itemize}
			\item [$\mathrm{(b')}$] $\E_Q[{L_i}\mid y_1,\ldots,y_{i-1}]\leqslant \mathrm{maxkl}(\diffp_i)$.
		\end{itemize}
		Combining (a) and $\mathrm{(b')}$, we can use $U_i(\lambda) = \frac{1}{8}\diffp_i^2\lambda^2+\lambda\cdot \mathrm{maxkl}(\diffp_i)$, which we label as ``KL-improved DR19" in \Cref{fig:U}.  This observation on the expectation together with the \citet{DurfeeRo19} bound that uses Azuma-Hoeffding, but with a weaker bound on the expectation term, we have the following result.
		
		\mgf*
		
		Instead of trying to come up with analytic, closed form upper bounds, we directly compute $\E_Q[e^{\lambda L_i}\mid y_1,\ldots,y_{i-1}]$, resorting to numerical tools when necessary. 
Recall that $p_{\diffp,t} = \frac{e^{-t}-e^{-\diffp}}{1-e^{-\diffp}}$ and $q_{\diffp,t} = e^tp_{\diffp,t} = \frac{1-e^{t-\diffp}}{1-e^{-\diffp}}$, we then have the following result. 
\begin{lemma} \label{lem:MGFexpression}
\begin{align*}
	\E_Q[e^{\lambda L_i}\mid y_1,\ldots,y_{i-1}]
	&=p_{\diffp_i,\diffp_i-t_i}^{\lambda+1} q_{\diffp_i,\diffp_i-t_i}^{-\lambda}+(1-p_{\diffp_i,\diffp_i-t_i})^{\lambda+1} (1-q_{\diffp_i,\diffp_i-t_i})^{-\lambda}
\end{align*}
where $t_i = t_i(y_{1}, \cdots, y_{i-1})$.
\end{lemma}
\begin{proof}
	Let $P_i$ be the distribution for $\mathrm{Bern}(p_{\diffp_i,t_i})$ and  $Q_i$ be the distribution for  $\mathrm{Bern}(q_{\diffp_i,t_i})$. Then
	\begin{align*}
		\E_Q[e^{\lambda L_i}\mid y_1,\ldots,y_{i-1}]
		&= \int \Big(\frac{Q_i}{P_i}\Big)^\lambda \cdot Q_i\\
		&= q_{\diffp_i,t_i}^{\lambda+1}p_{\diffp_i,t_i}^{-\lambda}+(1-q_{\diffp_i,t_i})^{\lambda+1}(1-p_{\diffp_i,t_i})^{-\lambda}
	\end{align*}
	It is easy to verify that $q_{\diffp,\diffp-t} = 1-p_{\diffp,t}$ and $p_{\diffp,\diffp-t} = 1-q_{\diffp,t}$. Plugging these into the above expression yields the desired result.
\end{proof}

We now simplify the expression in Lemma~\ref{lem:MGFexpression} with the following function,
\begin{align*}
	h_\ep(\lambda)
	&=\sup_{t\in[0,\ep]} \log \big(p_{\ep,t}^{\lambda+1} q_{\ep,t}^{-\lambda}+(1-p_{\ep,t})^{\lambda+1} (1-q_{\ep,t})^{-\lambda}\big)\\
	&= \sup_{t\in[0,\ep]} \log \big(p_{\ep,t}e^{-\lambda t}+(1-p_{\ep,t})e^{-\lambda (t-\ep)}\big)\\
	&=\sup_{t\in[0,\ep]}\lambda(\ep-t)+\log\big(1+p_{\ep,t}(e^{-\lambda\ep}-1)\big).
\end{align*}
The second line above makes use of the fact that
\[q_{\ep,t} = e^t p_{\ep,t} \text{ and } 1-q_{\ep,t} = e^{t-\ep} (1-p_{\ep,t}).\]
Now it is easy to see that $U_i(\lambda)$ can be taken as $h_{\diffp_i}(\lambda)$, which we label as ``General MGF" in \Cref{fig:U}. Combining this with \Cref{lem:origin}, we have \Cref{thm:MGFthm}.

\begin{figure}[!htp]
\centering
\includegraphics[width=0.7\textwidth]{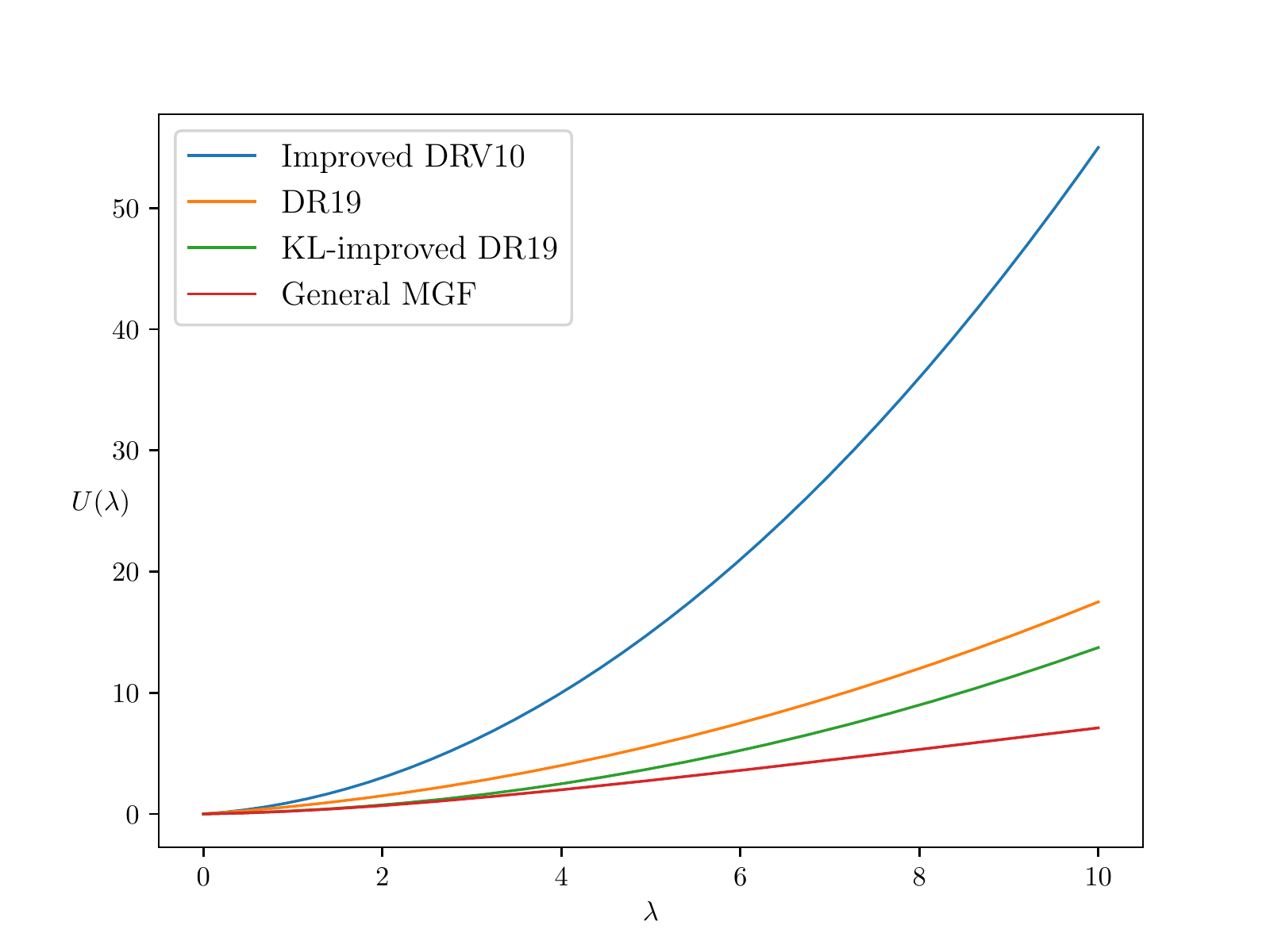}
\caption{A unified view and comparison of composition theorems involving concentration inequalities. The figure shows graphs of different $U$ functions (see \Cref{lem:origin}) used in different results, such as from \citet{DworkRoVa10} (labeled ``Improved DRV10'') and \citet{DurfeeRo19} (labeled ``DR19''). According to \Cref{lem:origin}, smaller function $U$ yields tighter privacy result. \Cref{thm:MGFthm} uses the smallest $U$ (labeled ``General MGF'') among all and is hence the tightest. All curves use $\ep=1$.\label{fig:U}}
\end{figure}

\paragraph{Numerical Issue} 
\label{par:numerical_issue}
We now point out a potential numeric issue in computing the function $h_\ep(\lambda)$.  Note that it can be simplified differently as
\begin{align*}
	h_\ep(\lambda)
	&= \sup_{t\in[0,\ep]} -\lambda t+\log \big(p_{\ep,t}+e^{\ep\lambda}(1-p_{\ep,t})\big).
\end{align*}
For comparision, the expression we use in \Cref{thm:MGFthm} is
\begin{align*}
	h_\ep(\lambda)
	&=\sup_{t\in[0,\ep]}\lambda(\ep-t)+\log\big(1+p_{\ep,t}(e^{-\lambda\ep}-1)\big).
\end{align*}
At first glance it may appear that the above two expressions are equal. However, the one used in the theorem is far more robust numerically, as in the optimization step, $\diffp\lambda$ can be large, which could make $e^{\diffp\lambda}$ beyond the range of floating point numbers.

%% file: conclusion.tex
\section{Conclusion and future directions} 

In this work, we studied the privacy loss when composing multiple exponential mechanisms, which is a fundamental class of DP algorithms.  We considered the privacy loss bounds when the exponential mechanisms can be adaptively selected at each round or when they are all selected in advance, as well as differentiated the homogeneous (all privacy parameters are the same) and the heterogeneous (privacy parameters can be different) case. We then made the connection between exponential mechanisms and the generalized randomized response mechanism to help simplify our privacy loss expressions.  Although we provided formulas for each case, we only provided an efficient calculation for computing the optimal composition bound in the nonadaptive and homogeneous case.  We conjecture that computing the optimal composition bound in the nonadaptive and heterogeneous case has similar hardness results as shown in \citet{MurtaghVa16} and we leave the problem open for future work.  

We then showed for the optimal homogenous composition bound that there is a separation between in the adaptive and nonadaptive case, which to our knowledge is a first of its kind result.  We then provided improved and computationally efficient composition bounds for the adaptive and homogeneous case by tailoring concentration bounds for our particular setting.  In order to better understand the adaptive composition bound, one potential direction for future work is to understand the asymptotics of the privacy loss bound, as $k \to \infty$.  We conjecture that the asymptotic gap collapses between the optimal composition bound for the adaptive and nonadaptive cases, and leave that as future work to study. 
Furthermore, in the non-asymptotic setting we believe that the gap between adaptive and non-adaptive is quite small, and also leave proving a strong upper bound on this gap to future work. 

Lastly, it is interesting to study composition bounds that account for different types of DP mechanisms at each round.  General DP composition bounds can be used in cases where Laplace and exponential mechanisms are used, but perhaps those bounds can be improved with composition that accounts for exponential mechanisms and Laplace mechanisms separately.  We leave this as an interesting direction of future work.  

%% file: acks.tex
\section{Acknowledgements} 

We thank our colleagues Reza Hosseini, Krishnaram Kenthapadi, Sean Peng, and Subbu Subramaniam for their helpful feedback and comments.

%% file: gen-rand-response.tex
\section{Proof of Lemma~\ref{lem:reduc}} 
\label{sec:gen-rand-response}

In order to use this interpretation, we will need to first establish some notation. For a pair of probability distributions $P$ and $Q$ on a common probability space $\Omega$, its trade-off function \cite{DongRoSu19} describes the hardness of the hypothesis testing problem $H_0: P$ vs $H_1: Q$. Let $E \subseteq \Omega$ be an arbitrary rejection region and
	\begin{align*}
		\alpha_E &= P[E]\\
		\beta_E &= 1-Q[E]
	\end{align*}
	be the type I and type II errors of the test $E$ respectively. Fix a level $\alpha_0$ and let $E$ run over all test with type I error at most $\alpha_0$, the minimal type II error is
	\[\inf\{\beta_E: E \text{ is a rejection region s.t. } \alpha_E\leqslant \alpha_0\}.\]
	This correspondence of $\alpha_0$ to the minimal type II error defines a function from $[0,1]$ to $[0,1]$. We will call this function $T(P,Q)$. Formally,
	\begin{align*}
		T(P,Q):[0,1]&\to[0,1]\\
		\alpha_0 &\mapsto \inf\{\beta_E: \alpha_E\leqslant \alpha_0\}
	\end{align*}
	
For our proof, we will use this function $T$ and apply Blackwell's theorem (\cite{blackwell1950comparison}, Theorem 10). The following form is taken from \cite{DongRoSu19}.
	\begin{theorem}\label{thm:blackwell}
	Let $P,Q$ be probability distributions on $Y$ and $P',Q'$ be probability distributions on $Z$. The following two statements are equivalent:
	\begin{enumerate}
	\item[(a)] $\F(P,Q)\leqslant \F(P',Q')$.
	\item[(b)] There exists a randomized algorithm $\Proc: Y \to Z$ such that $\Proc(P)=P',\Proc(Q)=Q'$.
	\end{enumerate}
	\end{theorem}

We now prove that we can post-process the generalized random response to simulate any BR mechanism on neighboring inputs.

\begin{proof}[Proof of \Cref{lem:reduc}]

Let $P$ be the outcome distribution of $M(x^0)$ and $Q$ be the outcome distribution of $M(x^1)$.
	By \Cref{cor:br},  we know there exists some $t \in [0,\ep]$ such that 
	\[t-\ep\leqslant\log\frac{Q(y)}{P(y)} \leqslant t.\]
	Equivalently, for any event $E \subseteq \cY$, 
	\begin{equation}\label{eq:PEQE1}
		\e^{t-\ep}P[E]\leqslant Q[E]\leqslant \e^tP[E].
	\end{equation}
	Applying the same rule for the complement event $E^c$, we have
	\begin{equation}\label{eq:PEQE2}
		\e^{t-\ep}P[E^c]\leqslant Q[E^c]\leqslant \e^tP[E^c].
	\end{equation}
	The second inequality of \eqref{eq:PEQE1} and the first inequality of \eqref{eq:PEQE2} imply
	\begin{equation}\label{eq:aebe}
		1-\beta_E\leqslant \e^t \alpha_E, \quad \e^{t-\ep}(1-\alpha_E)\leqslant \beta_E.
	\end{equation}
	Let the piece-wise linear function $l_{t,\ep}:[0,1]\to[0,1]$ be defined as
	\[
	l_{t,\ep}(x) = \max\{1-\e^tx, \e^{t-\ep}(1-x)\}.
	\]
	It's easy to see that \eqref{eq:aebe} implies $T(P,Q)\geqslant l_{t,\ep}$ pointwise in $[0,1]$. 

	Furthermore, it is straightforward to verify that $l_{t,\ep} \equiv T(\grr{\diffp,t}(0),\grr{\diffp,t}(1))$ because the respective inequalities in \eqref{eq:aebe} are tight for $E = \{0\}$ and $E = \{1\}$, respectively. Therefore, there must be a $t=t(M,x^0,x^1)$ such that
	\[
	T\big(M(x^0), M(x^1)\big)\geqslant T\big(\grr{\diffp,t}(0), \grr{\diffp,t}(1)\big).
	\]
Applying Theorem~\ref{thm:blackwell} then gives our desired claim.
	\end{proof}

%% file: non-interactive-proofs.tex
\section{Omitted Proofs from Section~\ref{sec:nonadaptive}}\label{sec:nonadaptiveapp} 

We provide here the proofs from Section~\ref{sec:nonadaptive} that were omitted.

\subsection{Proof of Lemma~\ref{lem:gen_convexity}}

This lemma will be proven in two main sublemmas. First, we show that it holds for $k = 2$, then we show how we can reduce the general case to $k = 2$ by conditioning outcomes other than the first and second terms.

\begin{lemma}\label{lem:2_convexity} For any $\diffp_g \in \R$ and $t_1, t_2 \in [0,\diffp]$
\[
\delta((t_1,t_2),\diffp_g) \leq \delta\left( \left(\frac{t_1 + t_2}{2}, \frac{t_1 + t_2}{2} \right), \diffp_g\right)
\]
Further, the inequality is strict whenever $\diffp_g< t_1 + t_2 <\diffp_g + 2\diffp$ and $t_1 \neq t_2$.
\end{lemma}
\begin{proof}
Using the fact that $q_{t} = e^t p_t$ and $(1 - q_t) = e^{t-\diffp} (1 - p_t)$, we rewrite 
\[
\delta((t_1,t_2),\diffp_g) = \sum_{S \subseteq \{1,2\}} \prod_{i \notin S}p_{t_i} \prod_{i \in S} (1 - p_{t_i}) \max\left\{e^{t_1 + t_2 - |S|\diffp}- e^{\diffp_g}, 0 \right\}
\]

We will then prove our desired inequality by considering four cases. 

\paragraph{Case I ($t_1 + t_2 \leq \diffp_g$):} This implies that $\max\{e^{t_1 + t_2 - |S|\diffp}- e^{\diffp_g}, 0 \} = 0$ for any subset $S$ and 

\[
\delta((t_1,t_2),\diffp_g) = \delta\left( \left(\frac{t_1+t_2}{2},\frac{t_1 + t_2}{2} \right),\diffp_g\right) = 0.
\]

\paragraph{Case II ($t_1 + t_2 \geq \diffp_g + 2\diffp$):} This implies $\max\{e^{t_1 + t_2 - |S|\diffp}- e^{\diffp_g}, 0 \} = e^{t_1 + t_2 - |S|\diffp}- e^{\diffp_g}$ for any $S$, which gives

\[
\delta((t_1,t_2),\diffp_g) = \sum_{S \subseteq \{1,2\}} \left(\prod_{i \notin S} q_{t_i} \prod_{i \in S} (1 - q_{t_i}) - e^{\diffp_g}\prod_{i \notin S}p_{t_i} \prod_{i \in S} (1 - p_{t_i}) \right)= 1 - e^{\diffp_g}
\]
and equivalently holds for $ \delta(\left(\frac{t_1+t_2}{2},\frac{t_1 + t_2}{2}\right),\diffp_g)$.

\paragraph{Case III ($\diffp_g < t_1 + t_2 \leq \diffp_g + \diffp$):} This implies that $\max\{e^{t_1 + t_2 - |S|\diffp}- e^{\diffp_g}, 0 \} = 0$ for any $S$ such that $|S| > 0$. Therefore,
\[
\delta((t_1,t_2),\diffp_g) = p_{t_1}p_{t_2}\left(e^{t_1 + t_2} - e^{\diffp_g}\right)
\]

Equivalently, we have 
\[
 \delta\left(\left( \frac{t_1+t_2}{2},\frac{t_1 + t_2}{2}\right),\diffp_g\right) = p_{\frac{t_1 + t_2}{2}}^2 \left(e^{t_1 + t_2} - e^{\diffp_g}\right)
\]

We want strict inequality for this case, so it suffices to show $p_{t_1}p_{t_2} < p_{\frac{t_1 + t_2}{2}}^2$. Plugging in the explicit formula for each $p_t$ and performing some simple algebraic manipulations gives that this is equivalent to
\[
2e^{-\frac{t_1 + t_2}{2}} < e^{-t_1} + e^{-t_2}
\]
which holds due to the strict-convexity of the exponential function.

\paragraph{Case IV ($\diffp_g + \diffp \leq t_1 + t_2 < \diffp_g + 2\diffp$):} This implies that $\max\{e^{t_1 + t_2 - |S|\diffp}- e^{\diffp_g}, 0 \} = 0$ when $|S| = 2$. Therefore,  

\[
\delta((t_1,t_2),\diffp_g) = p_{t_1}p_{t_2}\left(e^{t_1 + t_2} - e^{\diffp_g}\right) + \left(p_{t_1} (1 - p_{t_2}) + p_{t_2}(1 - p_{t_1})\right)\left(e^{t_1 + t_2 - \diffp} - e^{\diffp_g}\right)
\]
From Case II, we know
\[
\sum_{S \subseteq \{1,2\}} \prod_{i \notin S}p_{t_i} \prod_{i \in S} (1 - p_{t_i}) \left(e^{t_1 + t_2 - |S|\diffp}- e^{\diffp_g}\right) = 1 - e^{\diffp_g}
\]
which yields
\[
\delta((t_1,t_2),\diffp_g) = 1 - e^{\diffp_g} - (1 - p_{t_1})(1 - p_{t_2})\left(e^{t_1 + t_2 - 2\diffp} - e^{\diffp_g}\right).
\]
This equivalently holds for $ \delta(\left(\tfrac{t_1+t_2}{2},\tfrac{t_1 + t_2}{2}\right),\diffp_g)$ and because $e^{t_1 + t_2 - 2\diffp} - e^{\diffp_g} < 0$, we have
\[
\delta(\left(t_1,t_2\right),\diffp_g)  <  \delta\left(\left(\frac{t_1+t_2}{2},\frac{t_1 + t_2}{2}\right),\diffp_g\right) \qquad \Leftrightarrow \qquad 
(1 - p_{t_1})(1 - p_{t_2})  <  \left(1 - p_{\frac{t_1 + t_2}{2}}\right)^2.
\]

Once again, we plug in the explicit formula for each $p_t$ and perform some simple algebraic manipulations to see that this is also equivalent to

\[
2e^{-\frac{t_1 + t_2}{2}} < e^{-t_1} + e^{-t_2}
\]

and this again holds due to the strict-convexity of the exponential function.

\end{proof}

We now want to extend this to $k > 2$, which will be done by fixing an arbitrary subset of $\{3,\cdots,k\}$ and show that the inequality holds when we restrict the summation to subsets of $\{1, \cdots,k\}$ that must contain that subset of $\{3,\cdots,k\}$. This will allow for easy cancellation.  We will denote $\delta_U(\bbt,\diffp_g,S)$ for a set $U \subseteq [k]$ and $S \subseteq U$ as 

\begin{multline*}
\delta_U(\bbt,\diffp_g,S)\defeq \prod_{i \in U\setminus S} p_{t_i} \prod_{i \in S} (1 - p_{t_i}) \\
\cdot   \sum_{S' \subseteq [k] \setminus U} \max\left\{ e^{\sum_{j \in U}t_j - |S|\diffp} \prod_{i \notin U \cup S'} q_{t_i} \prod_{i \in S'} (1 - q_{t_i}) - e^{\diffp_g} \prod_{i \notin U \cup S'}p_{t_i} \prod_{i \in S'} (1 - p_{t_i}), 0 \right\}.
\end{multline*}

\begin{claim}\label{claim:sum_fixed_sets}
Let $\diffp_g \in \R$.  Then for any $\bbt \in [0,\ep]^k$, we have for $U = \{3, \cdots, k \}$
\[
\delta(\bbt,\diffp_g) = \sum_{S \subseteq U }  \delta_U(\bbt,\diffp_g,S) 
\]

\end{claim}
\begin{proof}
We fix a set $S \subseteq \{3, \cdots, k \} = U$.  Using the fact that $q_{t} = e^t p_t$ and $(1 - q_t) = e^{t-\diffp} (1 - p_t)$, we have
\[
\prod_{i \in U\setminus S} q_{t_i} \prod_{i \in S} (1 - q_{t_i}) = e^{t_3 + \cdots + t_k - |S |\diffp} \prod_{i \in U\setminus S} p_{t_i} \prod_{i \in S} (1 - p_{t_i}) 
\]
Therefore, we also have 
\[
\delta_U(\bbt,\diffp_g,S) = \sum_{S' \subseteq \{1,2\}} \max\left\{ \prod_{i \notin S' \cup S} q_{t_i} \prod_{i \in S' \cup S} (1 - q_{t_i}) - e^{\diffp_g} \prod_{i \notin S' \cup S}p_{t_i} \prod_{i \in S' \cup S} (1 - p_{t_i}), 0 \right\}
\]
Summing over all $S$ we can simply rewrite this summation over all subsets of $\{1,\cdots,k\}$, giving our desired equality.
\end{proof}

\begin{lemma}\label{lem:set_fixed_convexity}
For any $S \subseteq \{3,...,k\} = U$, we have the following inequality
\[
\delta_U(\bbt,\diffp_g,S ) \leq \delta_U\left(\left(\frac{t_1 + t_2}{2}, \frac{t_1 + t_2}{2}, t_3,...,t_k\right),\diffp_g, S\right)
\]
Further, the inequality is strict if $\diffp_g < \sum_{i=1}^k t_i - |S|\diffp <\diffp_g + 2 \diffp$ and $t_1 \neq t_2$.
\end{lemma}
\begin{proof}
We fix $S \subseteq \{3,\cdots, k \}$.  Let $\diffp'_g = \diffp_g + |S| - t_3 - \cdots - t_k$, and then by cancelling non-negative like terms it suffices to show

\begin{multline*}
\sum_{S' \subseteq \{1,2\}} \max\left\{ \prod_{i \notin S'} q_{t_i} \prod_{i \in S'}(1 - q_{t_i}) - e^{\diffp'_g} \prod_{i \notin S'} p_{t_i} \prod_{i \in S'} (1 - p_{t_i}), 0 \right \} 
\\
\leq
\sum_{S' \subseteq \{1,2\}} \max\left\{ \prod_{i \notin S'} q_{t'} \prod_{i \in S'}(1 - q_{t'}) - e^{\diffp'_g} \prod_{i \notin S'} p_{t'} \prod_{i \in S'} (1 - p_{t'}), 0 \right \} 
\end{multline*}
where $t' = \frac{t_1 + t_2}{2}$. By definition, this is then equivalent to showing
\[
\delta( (t_1,t_2),\diffp'_g) \leq \delta\left(\left(\frac{t_1 + t_2}{2}, \frac{t_1 + t_2}{2}\right), \diffp'_g\right)
\]
which follows from Lemma~\ref{lem:2_convexity}, and the strictness follows from the fact that $\diffp'_g = \diffp_g + |S| - \sum_{j > 2} t_j$.
\end{proof}

With these we can now prove our main convexity lemma.

\begin{proof}[Proof of Lemma~\ref{lem:gen_convexity}]
It immediately follows from Claim~\ref{claim:sum_fixed_sets} and Lemma~\ref{lem:set_fixed_convexity} that for any $\bbt \in [0,\ep]^k$

\[
\delta(\bbt,\diffp_g) \leq \delta\left(\left(\frac{t_1 + t_2}{2}, \frac{t_1 + t_2}{2}, t_3,...,t_k\right),\diffp_g\right)
\]

Additionally, if we assume that $t_1 \neq t_2$ and $\diffp_g <\sum t_i < \diffp_g + k \diffp$, then there must exist some $\ell \in [0,k-2]$ such that $\diffp_g + \ell\diffp< \sum t_i < \diffp_g + (\ell + 2)\diffp$, which implies that $\diffp_g< \sum t_i - \ell \diffp < \diffp_g  + 2\diffp$. Further,
we know that for any $\ell \in [0,k-2]$ there exists $S \subseteq \{3, \cdots, k\}$ such that $|S| = \ell$.
Therefore, for one of these subsets the inequality is strict and the sum must be a strict inequality as well.
\end{proof}

\subsection{Proof of Lemma~\ref{lem:partial_derivative_full}}

Recall that we had the following definition, for which we wanted to compute the partial derivate with respect to $t$.

\begin{equation}
F_{\ell}(t,\diffp_g) \defeq \sum_{i = 0}^{\ell} {k \choose i} p_{t}^{k-i}(1 - p_{t})^{i} \left( e^{kt - i\diffp} - e^{\diffp_g} \right)
\label{eq:F_t}
\end{equation}

We further split each $F_{\ell}(t,\diffp_g)$ into the individual terms to more easily differentiate the full summation with respect to $t$.
\[
f_{\ell}(t,\diffp_g) \defeq  {k \choose \ell} p_{t}^{k-\ell}(1 - p_{t})^{\ell} \left( e^{kt - \ell\diffp} - e^{\diffp_g} \right)
\]

In particular, giving a much simpler formulation for the partial derivative will rely upon an inductive proof, so this definition will allow an even easier comparison between $F_{\ell}(t,\diffp_g)$ and $F_{\ell + 1}(t,\diffp_g)$ that follows immediately from the definition.

\begin{corollary}\label{cor:inductive_for_t} For any $\ell \in [1,k]$
\[
F_{\ell}(t,\diffp_g) = F_{\ell - 1}(t,\diffp_g) +  f_{\ell}(t,\diffp_g)
\]
\end{corollary}

We first differentiate the simplest of these expressions $F_0(t,\diffp_g)$, and then we will ultimately use this as the base case for proving a simplified formulation of derivative for the general case.

\begin{lemma}\label{lem:base_case}

\[
\frac{\partial F_0(t,\diffp_g)}{\partial t}  = k p_{t}^{k-1} \frac{1}{1 - e^{-\diffp}} \left(e^{\diffp_g - t} - e^{kt - \diffp}\right)
\]

\end{lemma}

\begin{proof}
By definition

\[
F_0(t,\diffp_g) = p_t^k \left( e^{kt} - e^{\diffp_g}\right) = \left( \frac{e^{-t} - e^{-\diffp}}{1 - e^{-\diffp}} \right)^k \left( e^{kt} - e^{\diffp_g}\right)
\]

Therefore, by basic differentiation rules

\begin{multline*}
\frac{\partial F_0(t,\diffp_g)}{\partial t} = \left(-k\frac{e^{-t}}{1 - e^{-\diffp}} \left( \frac{e^{-t} - e^{-\diffp}}{1 - e^{-\diffp}} \right)^{k-1} \left(e^{kt} - e^{\diffp_g}\right)\right) + \left(\frac{e^{-t} - e^{-\diffp}}{1 - e^{-\diffp}}\right)^k k e^{kt} 
\\
= k p_t^{k-1} \frac{1}{1 - e^{-\diffp}} \left( - e^{-t}\left(e^{kt} - e^{\diffp_g} \right) + \left(e^{-t} - e^{\diffp} \right) e^{kt}  \right) 
\end{multline*}
which easily reduces to our desired term.

\end{proof}

To apply an inductive claim to the general case, we will also need to evaluate the partial derivative of the last term for each sum.

\begin{lemma}\label{lem:inductive_step}
For $1 \leq \ell \leq k$ 

\begin{multline*}
\frac{\partial f_{\ell}(t,\diffp_g)}{\partial t} = {k \choose \ell} p_{t}^{k - 1 - \ell} (1 - p_{t})^{\ell - 1} \left(\frac{1}{1 - e^{-\diffp}} \right)^2 \bigg( (k - \ell)\left(e^{\diffp_g - t} + e^{(k-1)t - (\ell + 1)\diffp}\right) \\
+ \ell\left(e^{(k - 1)t - \ell\diffp} + e^{\diffp_g - \diffp - t}\right) - k\left(e^{\diffp_g - 2t} + e^{kt - (\ell + 1)\diffp}\right)\bigg) 
\end{multline*}

\end{lemma}

\begin{proof}
By definition

\[
f_{\ell}(t,\diffp_g) = {k \choose \ell} p_{t}^{k-\ell}(1 - p_{t})^{\ell} \left( e^{kt - \ell\diffp} - e^{\diffp_g} \right)
\]

We can consider this then to instead be $f_{\ell}(t,\diffp_g) = {k \choose \ell} f(t) \cdot g(t) \cdot h(t)$ with $f(t) =  p_{t}^{k-\ell}$, $g(t) = (1 - p_{t})^{\ell}$, and $h(t) = e^{kt - \ell\diffp} - e^{\diffp_g}$. Applying basic differentiation rules and using the fact that $p_t = \frac{e^{-t} - e^{-\diffp}}{1 - e^{-\diffp}}$, we obtain

\begin{multline*}
\frac{\partial f_{\ell}(t,\diffp_g)}{\partial t} = {k \choose \ell}( k - \ell) \left(\frac{-e^{-t}}{1 - e^{-\diffp}}\right) p_{t}^{k - 1 - \ell} (1 - p_{t})^{\ell } \left(e^{kt - \ell \diffp} - e^{\diffp_g}\right)
\\
+ {k \choose \ell} \ell \left( \frac{e^{-t}}{1 - e^{-\diffp}} \right) p_{t}^{k - \ell} (1 - p_{t})^{\ell - 1} \left(e^{kt - \ell \diffp} - e^{\diffp_g}\right)
+ {k \choose \ell} k e^{kt - \ell\diffp} p_{t}^{k  - \ell} (1 - p_{t})^{\ell } 
\end{multline*}

We can pull out similar terms from each expression to achieve

\begin{multline*}
\frac{\partial f_{\ell}(t,\diffp_g)}{\partial t} = {k \choose \ell} p_{t}^{k - 1 - \ell} (1 - p_{t})^{\ell - 1} \left( \frac{1}{1 - e^{-\diffp}} \right)^2
\bigg( 
-(k - \ell)e^{-t}(1 - e^{-t}) \left(e^{kt - \ell \diffp} - e^{\diffp_g}\right) 
\\
+ \ell e^{-t}(e^{-t} - e^{-\diffp})\left(e^{kt - \ell \diffp} - e^{\diffp_g}\right) + 
ke^{kt - \ell\diffp}\left(e^{-t} - e^{-\diffp}\right)\left(1 - e^{-t}\right)
\bigg)
\end{multline*}

Further examination of the inner term by expanding each expression and cancelling like terms gives

\begin{multline*}
-(k - \ell)e^{-t}(1 - e^{-t}) \left(e^{kt - \ell \diffp} - e^{\diffp_g}\right) +
\ell e^{-t}(e^{-t} - e^{-\diffp})\left(e^{kt - \ell \diffp} - e^{\diffp_g}\right) + 
ke^{kt - \ell\diffp}\left(e^{-t} - e^{-\diffp}\right)\left(1 - e^{-t}\right) 
\\
=
  (k - \ell)\left(e^{\diffp_g - t} + e^{(k-1)t - (\ell + 1)\diffp}\right) 
+ \ell\left(e^{(k - 1)t - \ell\diffp} + e^{\diffp_g - \diffp - t}\right) - k\left(e^{\diffp_g - 2t} + e^{kt - (\ell + 1)\diffp}\right) 
\end{multline*}

This then implies our desired expression.

\end{proof}

We now have the pieces to give a simpler evaluation of the partial derivative for the general case using an inductive argument. Surprisingly, with a bit of combinatorial and algebraic massaging, the full partial derivative will reduce to a rather simple expression.

\begin{proof}[Proof of Lemma~\ref{lem:partial_derivative_full}]
The base case of $\ell = 0$ is true from Lemma~\ref{lem:base_case}. We then assume the claim for $\ell - 1$, and by Corollary~\ref{cor:inductive_for_t} we know $F_{\ell}(t,\diffp_g) = F_{\ell - 1}(t,\diffp_g) + f_{\ell}(t,\diffp_g)$, which implies 

\[
\frac{\partial F_{\ell}(t,\diffp_g)}{\partial t} = \frac{\partial F_{\ell-1}(t,\diffp_g)}{\partial t} + \frac{\partial f_{\ell}(t,\diffp_g)}{\partial t}
\]

Applying our inductive claim and Lemma~\ref{lem:inductive_step} we then have 

\begin{multline*}
\frac{\partial F_{\ell}(t,\diffp_g)}{\partial t} = (k - (\ell - 1)) {k \choose \ell - 1} p_{t}^{k - 1 - (\ell - 1)} (1 - p_{t})^{\ell - 1} \frac{1}{1 - e^{-\diffp}} \left(e^{\diffp_g - t} - e^{kt - \ell\diffp} \right) + 
\\
{k \choose \ell} p_{t}^{k - 1 - \ell} (1 - p_{t})^{\ell - 1} \left(\frac{1}{1 - e^{-\diffp}} \right)^2 \bigg( (k - \ell)\left(e^{\diffp_g - t} + e^{(k-1)t - (\ell + 1)\diffp}\right) \\
+ \ell\left(e^{(k - 1)t - \ell\diffp} + e^{\diffp_g - \diffp - t}\right) - k\left(e^{\diffp_g - 2t} + e^{kt - (\ell + 1)\diffp}\right)\bigg) 
\end{multline*}

We use the fact that $(k - (\ell - 1)){k \choose \ell -1} = \ell{k \choose \ell}$ and this reduces to 

\begin{multline*}
\frac{\partial F_{\ell}(t,\diffp_g)}{\partial t} = 
{k \choose \ell} p_{t}^{k - 1 - \ell} (1 - p_{t})^{\ell - 1}  \left(\frac{1}{1 - e^{-\diffp}} \right)^2 \bigg( \ell \left(e^{-t} - e^{-\diffp}\right) \left(e^{\diffp_g - t} - e^{kt - \ell\diffp} \right) +
\\
(k - \ell)\left(e^{\diffp_g - t} + e^{(k-1)t - (\ell + 1)\diffp}\right) 
+ \ell\left(e^{(k - 1)t - \ell\diffp} + e^{\diffp_g - \diffp - t}\right) - k\left(e^{\diffp_g - 2t} + e^{kt - (\ell + 1)\diffp}\right)\bigg) 
\end{multline*}

Further examination of the inner term by expanding each expression and cancelling like terms gives

\begin{multline*}
 \ell \left(e^{-t} - e^{-\diffp}\right) \left(e^{\diffp_g - t} - e^{kt - \ell\diffp} \right) +
(k - \ell)\left(e^{\diffp_g - t} + e^{(k-1)t - (\ell + 1)\diffp}\right) 
\\
+ \ell\left(e^{(k - 1)t - \ell\diffp} + e^{\diffp_g - \diffp - t}\right) - k\left(e^{\diffp_g - 2t} + e^{kt - (\ell + 1)\diffp}\right)
\\
=
(k - \ell) \left(e^{\diffp_g - t} - e^{\diffp_g - 2t} + e^{(k-1)t - (\ell + 1)\diffp} - e^{kt - (\ell + 1)\diffp}\right)
\\
=
(k - \ell) (1 - e^{-t}) \left(e^{\diffp_g - t} -  e^{kt - (\ell + 1)\diffp}\right)
\end{multline*}

Substituting for this simplified expression and using the fact that $1 - p_t = \frac{1 - e^{-t}}{1 - e^{-\diffp}}$ then gives our desired result.
\end{proof}

%% file: gap-proofs.tex
\section{Omitted Proofs from Section~\ref{sec:gap}}\label{sec:gap-proofs} 

We provide here the proofs from Section~\ref{sec:gap} that were omitted.

\subsection{Proofs from Section~\ref{subsec:gap}}

\begin{proof}[Proof of Corollary~\ref{cor:non_adap_recursion}]
Note that by our definition, $q_t = e^t p_t$ and $1 - q_t = e^{t-\diffp}(1 - p_t)$, so we can equivalently write

\[
\delta^k(t,\diffp_g) = 
\sum_{i = 0}^k {k \choose i} q_{ t }^{k-i}(1 - q_{ t })^{i} \max\left\{\left( 1 - e^{\diffp_g - kt + i\diffp} \right), 0 \right\}.
\]

We then prove by induction. For $k = 1$, the base case, 
\[
\delta^1(t, \diffp_g) = q_t \max\{1 - e^{\diffp_g - t}, 0\} + (1 - q_t)\max\{1 - e^{\diffp_g - t + \diffp}, 0 \}, 
\]
and the claim follows by definition of $\delta^0(t, \diffp_g)$.
We can then apply our inductive hypothesis to get both
\begin{align*}
q_t \cdot \delta^{k-1}(\diffp_g - t) & = \sum_{i=0}^{k-1} {k-1 \choose i} q_{t}^{k-i} (1 - q_t)^i \max\left\{\left( 1 - e^{\diffp_g - kt + i\diffp} \right), 0 \right\}, 
\\
(1 - q_t)\cdot \delta^{k-1}(\diffp_g - t + \diffp) & = \sum_{i = 0}^{k-1} {k-1 \choose i} q_t^{k-1-i} (1 - q_t)^{i + 1} \max\left\{\left( 1 - e^{\diffp_g - kt + (i + 1)\diffp} \right), 0 \right\} 
\\
& = \sum_{i = 1}^{k} {k-1 \choose i - 1} q_t^{k-i} (1 - q_t)^{i } \max\left\{\left( 1 - e^{\diffp_g - kt + i \diffp} \right), 0 \right\} .
\end{align*}

Our claim then follows from the fact that for any $i \in [1,k-1]$, we must have ${k-1 \choose i -1} + {k-1 \choose i} = {k \choose i}$.
\end{proof}

\begin{proof}[Proof of Lemma~\ref{lem:recursive_gap}]

We prove this inductively. For the base case $k = 2$, from Corollary~\ref{cor:non_adap_recursion} and our definition of $t_{\opt}(\nonadaptBR^2,\diffp_g)$ we have

\[
\delta_{\opt}(\nonadaptBR^2,\diffp_g) = q_t \delta^{1}(t,\diffp_g - t) + (1 - q_t)\delta^{1}(t,\diffp_g + \diffp - t)
\]
If $t \notin t_{\opt}(\nonadaptBR^1,\diffp_g - t + \ell'\diffp)$ for some $\ell' \in \{0,1\}$, then $\delta_{\opt}(\nonadaptBR^1,\diffp_g - t + \ell'\diffp) > \delta^{1}(t,\diffp_g - t + \ell'\diffp)$. Applying Lemma~\ref{lem:adap_recursion_hetero} for the homogeneous case, 

\begin{multline*}
\delta_{\opt}(\abr^2,\diffp_g) \geq q_t\delta_{\opt}(\abr^1,\diffp_g - t) + (1 - q_t) \delta_{\opt}(\abr^1,\diffp_g - t + \diffp)
\\
> q_t \delta^{1}(t,\diffp_g - t) + (1 - q_t)\delta^{1}(t,\diffp_g + \diffp - t) = \delta_{\opt}(\nonadaptBR^2,\diffp_g)
\end{multline*}
This equivalently follows if $\delta_{\opt}(\nonadaptBR^1,\diffp_g - t + \ell'\diffp) < \delta_{\opt}(\abr^1,\diffp_g - t + \ell'\diffp) $ for either $\ell' \in \{0,1\}$.

The inductive step will then follow equivalently. Once again, we have 

\[
\delta_{\opt}(\nonadaptBR^k,\diffp_g) = q_t \delta^{k-1}(t,\diffp_g - t) + (1 - q_t)\delta^{k-1}(t,\diffp_g + \diffp - t)
\]
which similarly implies

\begin{multline*}
\delta_{\opt}(\abr^{k},\diffp_g) \geq q_t\delta_{\opt}(\abr^{k-1},\diffp_q - t) + (1 - q_t) \delta_{\opt}(\abr^{k-1},\diffp_g - t + \diffp)
\\
\geq q_t\delta_{\opt}(\nonadaptBR^{k-1},\diffp_q - t) + (1 - q_t) \delta_{\opt}(\nonadaptBR^{k-1},\diffp_g - t + \diffp)
\\
\geq q_t \delta^{k-1}(t,\diffp_g - t) + (1 - q_t)\delta^{k-1}(t,\diffp_g + \diffp - t) 
= \delta_{\opt}(\nonadaptBR^k,\diffp_g)
\end{multline*}
The goal will then be to show that this inequality becomes strict if one of the conditions in the statement holds. 
First, suppose $t \notin t_{\opt}(\nonadaptBR^{k-1},\diffp_g - t + \ell'\diffp)$ for either $\ell' \in \{0,1\}$, then $\delta_{\opt}(\nonadaptBR^{k-1},\diffp_g - t + \ell'\diffp) > \delta_g^{k-1}(t,\diffp_g - t + \ell'\diffp)$ and the inequality must be strict. On the other hand, if $t \in t_{\opt}(\nonadaptBR^{k-1},\diffp_g - t + \ell'\diffp)$ for both $\ell' \in \{0,1\}$, then this fits the condition of our inductive hypothesis, and we will then use this to prove our claim for the remaining cases.

Let $0\leq \ell' \leq \ell < k$ be such that  $\delta_{\opt}(\abr^{k-\ell},\diffp_g - \ell t + \ell'\diffp) > \delta_{\opt}(\nonadaptBR^{k-\ell},\diffp_g - \ell t + \ell'\diffp) $, and if $\ell = 0$, then the inequality holds trivially. 
If $\ell \geq 1$, then rewriting the inequality, we equivalently have both of the following inequalities, 
\begin{align*}
& \delta_{\opt}(\abr^{k-1-(\ell-1)},\diffp_g - t - (\ell - 1) t + \ell'\diffp)  > \delta_{\opt}(\nonadaptBR^{k-1 -(\ell - 1)},\diffp_g - t - (\ell - 1) t + \ell' \diffp), \\ 
\text{ and } & \quad \delta_{\opt}(\abr^{k-1-(\ell-1)},\diffp_g - t + \diffp - (\ell - 1) t + (\ell' - 1)\diffp)  \\
& \qquad\qquad\qquad\qquad\qquad\qquad\qquad\qquad\quad  > \delta_{\opt}(\nonadaptBR^{k-1 -(\ell - 1)},\diffp_g - t + \diffp - (\ell - 1) t + (\ell' - 1)\diffp).
\end{align*}
 If $\ell \geq 1$, then we must have either $0 \leq \ell'  \leq (\ell - 1) < k - 1$ or $0 \leq (\ell' - 1) \leq (\ell - 1) < k - 1$. 
 We can then apply our inductive hypothesis to achieve $\delta_{\opt}(\abr^{k-1},\diffp_g-t) > \delta_{\opt}(\nonadaptBR^{k-1},\diffp_g - t)$, or $\delta_{\opt}(\abr^{k-1},\diffp_g-t + \diffp) > \delta_{\opt}(\nonadaptBR^{k-1},\diffp_g - t + \diffp)$, respectively, which implies that our inequality is strict.
 
 Similarly, let $0\leq \ell' \leq \ell < k$ be such that  $t \notin t_{\opt}(\nonadaptBR^{k-\ell},\diffp_g - \ell t + \ell'\diffp)$
 By definition we cannot have $\ell = 0$, and we previously considered $\ell = 1$, so we assume $\ell > 1$ in order to apply our inductive claim.
 Rewriting the set $t_\opt$, we must then have both hold 
 \begin{align*}
& \qquad  t \notin t_{\opt}(\nonadaptBR^{k-1-(\ell - 1)},\diffp_g -t -  (\ell - 1) t + \ell'\diffp) \\
 \text{ and } & \qquad t \notin t_{\opt}(\nonadaptBR^{k-1 - (\ell - 1)},\diffp_g - t + \diffp - (\ell -  1) t + (\ell' - 1)\diffp).
 \end{align*} 
 If $\ell > 1$, then $\ell - 1 > 0$ and either $0 \leq \ell'  \leq (\ell - 1) < k - 1$ or $0 \leq (\ell' - 1) \leq (\ell - 1) < k - 1$.
 Applying our inductive hypothesis, we have either case hold, respectively
\begin{align*}
\delta_{\opt}(\abr^{k-1},\diffp_g-t) & > \delta_{\opt}(\nonadaptBR^{k-1},\diffp_g - t), \\
\text{ or }  \delta_{\opt}(\abr^{k-1},\diffp_g-t + \diffp) & > \delta_{\opt}(\nonadaptBR^{k-1},\diffp_g - t + \diffp).
\end{align*}
This implies our inequality is strict.
\end{proof}

In order to prove Lemma~\ref{lem:use_2_case}, we will also need the following edge case.

\begin{lemma}\label{lem:edge_2_case}
$t_{\opt}(\nonadaptBR^2,-3\diffp/2) \cap t_{\opt}(\nonadaptBR^2,\diffp/2) = \emptyset$ and $ t_{\opt}(\nonadaptBR^2,-\diffp/2)  \cap t_{\opt}(\nonadaptBR^2,\diffp/2) = \emptyset$
\end{lemma}

\begin{proof}
For any $\diffp_g \in \{-3\diffp/2,-\diffp/2,\diffp/2,3\diffp/2\}$, from Lemma~\ref{lem:set_of_opt_t} that $t_{\opt}(\nonadaptBR^2,\diffp_g) \subseteq \{\frac{\diffp_g + (\ell + 1)\diffp}{3}\} \cap (0,\diffp)$ for $\ell \in \{0,1\}$. This then implies that $t_{\opt}(\nonadaptBR^2,-3\diffp/2) = \diffp/6$, $t_{\opt}(\nonadaptBR^2,-\diffp/2) \subseteq \{\diffp/6, \diffp/2\}$, $t_{\opt}(\nonadaptBR^2,\diffp/2) \subseteq \{\diffp/2,5\diffp/6\}$, and $t_{\opt}(\nonadaptBR^2,3\diffp/2) = 5\diffp/6$. The claim then follows immediately.

\end{proof}

\begin{proof}[Proof of Lemma~\ref{lem:use_2_case}]
By our assumptions, it immediately follows that either there exists $0 \leq j \leq k-2$ such that $\diffp_g - (k-2)t + j\diffp \in (-\diffp/2,\diffp/2)$, or we are in the edge case where there exists $0 \leq j < k-2$ such that $\diffp_g - (k-2)t + j\diffp =-\diffp/2$. In first case, 
we know that $\delta_{\opt}(\nonadaptBR^2, \diffp_g - (k-2)t_{\ell} + j \diffp) < \delta_{\opt}(\abr^2, \diffp_g - (k-2)t + j \diffp)$ from Lemma~\ref{lem:gap_base_case}. In the second case (the edge case), if $j = 0$ then we know $j + 2 \leq k-2$ because $k \geq 4$, and from Lemma~\ref{lem:edge_2_case} we must either have
$t \notin t_{\opt}(\nonadaptBR^2, \diffp_g - (k-2)t_{\ell} + j \diffp)$  or $t \notin t_{\opt}(\nonadaptBR^2, \diffp_g - (k-2)t_{\ell} + (j + 2) \diffp)$. Otherwise, if $j > 0$, then we again have from Lemma~\ref{lem:edge_2_case} that either $t \notin t_{\opt}(\nonadaptBR^2, \diffp_g - (k-2)t_{\ell} + (j -1) \diffp)$ or $t \notin t_{\opt}(\nonadaptBR^2, \diffp_g - (k-2)t_{\ell} + (j +1) \diffp)$.

In either case, we can immediately apply Lemma~\ref{lem:recursive_gap} to achieve our desired inequality.
\end{proof}

\subsection{Proofs from Section~\ref{subsec:nogap}}

\begin{proof}[Proof of Lemma~\ref{lem:equal_in_trivial}]

We will prove both statements by induction, where the base case $\delta_{\opt}(\abr^0,\diffp_g) = \max\{ 1 - e^{\diffp_g}, 0\} = 0$ for $\diffp_g \geq 0$ and $\delta_{\opt}(\abr^0,\diffp_g) = \max\{ 1 - e^{\diffp_g}, 0\} = 1 - e^{\diffp_g}$ for $\diffp_g \leq 0$.
For any $t \in [0,\diffp]$, if $\diffp_g \geq k\diffp$ we must have $\diffp_g - t \geq (k-1)\diffp$ and $\diffp_g - t + \diffp \geq (k-1)\diffp$. Similarly, if $\diffp_g \leq -k\diffp$ we must have $\diffp_g - t \leq -(k-1)\diffp$ and $\diffp_g - t + \diffp \leq -(k-1)\diffp$. Using Lemma~\ref{lem:adap_recursion_hetero} for the homogeneous case, we know

\[
\delta_{\opt}(\abr^{k},\diffp_g) = \sup_{t \in [0,\diffp]} \left\{ q_{t} \delta_{\opt}(\abr^{k-1}, \diffp_g - t) + (1 - q_{t})  \delta_{\opt}(\abr^{k-1}, \diffp_g + \diffp - t) \right\}
\]
and applying our inductive hypothesis easily gives $\delta_{\opt}(\abr^{k},\diffp_g) = 0$ for any $\diffp_g \geq k\diffp$. Applying our inductive hypothesis for $\diffp_g \leq -k\diffp$, we have for any $t$ that 

\begin{multline*}
q_{t} \delta_{\opt}(\abr^{k-1}, \diffp_g - t) + (1 - q_{t})  \delta_{\opt}(\abr^{k-1}, \diffp_g + \diffp - t) 
\\
= \frac{1 - e^{t-\diffp}}{1 - e^{-\diffp}}\left(1 - e^{\diffp_g - t}\right) + \frac{e^{t-\diffp} - e^{-\diffp}}{1 - e^{-\diffp}}\left(1 - e^{\diffp_g - t + \diffp}\right) = 1 - e^{\diffp_g}.
\end{multline*}

This then implies $\delta_{\opt}(\abr^{k},\diffp_g) = 1 - e^{\diffp_g}$ for any $\diffp_g \leq -k\diffp$.

\end{proof}

\begin{proof}[Proof of Lemma~\ref{lem:lower_no_gap}]
We show this inductively. For $k = 1$, if $\diffp_g \leq 0$, then for any $t \in [0,\diffp]$, we must have $\diffp_g - t \leq 0$ and $\max\{1 - e^{\diffp_g - t},0\} = 1 - e^{\diffp_g - t}$. This then implies
\[
\delta_{\opt}(\abr^1,\diffp_g) =   \sup_{t \in [0,\diffp]}  \left\{q_{t} (1-e^{\diffp_g - t})  + (1 - q_{t})\max\{1 - e^{\diffp_g + \diffp - t},0\} \right\}.
\]

If $\diffp_g \leq -\diffp$ then $\max\{1 - e^{\diffp_g + \diffp - t}, 0\} = 1 - e^{\diffp_g + \diffp - t}$ for any $t \in [0,\diffp]$ 
and $\delta_{\opt}(\abr^1,\diffp_g) = 1 - e^{\diffp_g}$ because from the proof of 
Lemma~\ref{lem:equal_in_trivial} we have $q_{t} (1-e^{\diffp_g - t})  + (1 - q_{t})(1 - e^{\diffp_g + \diffp - t}) = 1 - e^{\diffp_g}$ for any $t$. 
Furthermore, $e^{\diffp_g + \diffp - t} - 1 \leq 0$ for any $t \in [0,\diffp]$, so $\sup_{t\in[0,\diffp]} \left\{ (1 - q_t)(e^{\diffp_g + \diffp - t} - 1) \right\}= 0$ by setting $t = 0$, and we have our desired equality.

If $\diffp <\diffp_g \leq 0$, then there must exist some $t \in [0,\diffp]$ such that $\diffp_g + \diffp - t > 0$.
Once again, we know $q_{t} (1-e^{\diffp_g - t})  + (1 - q_{t})(1 - e^{\diffp_g + \diffp - t}) = 1 - e^{\diffp_g}$ for any $t$. Consequently, the supremum must be achieved for some $t \in [0,\diffp_g + \diffp) \subset [0,\diffp]$ such that $1 - e^{\diffp_g + \diffp - t} < 0$. Thus,

\begin{align*}
\delta_{\opt}(\abr^1,\diffp_g) & =   \sup_{t \in [0,\diffp_g + \diffp)}  q_{t} (1-e^{\diffp_g - t})  
\\
& =  \sup_{t \in [0,\diffp_g + \diffp)}  \left\{q_{t} (1-e^{\diffp_g - t})  + (1 - q_{t})(1 - e^{\diffp_g + \diffp - t}) + (1 - q_{t})( e^{\diffp_g + \diffp - t} - 1) \right\} 
\\
& =1 - e^{\diffp_g} + \sup_{t \in [0,\diffp_g + \diffp)}  (1 - q_{t})( e^{\diffp_g + \diffp - t} - 1)  
\\
& = 1 - e^{\diffp_g} + \sup_{t \in [0,\diffp]}  (1 - q_{t})( e^{\diffp_g + \diffp - t} - 1) .
\end{align*}

The inductive step for $k \geq 2$ follows more easily, where for any $t \in [0,\diffp]$, we must have $\diffp_g - t \leq -(k-1)\diffp$, so from Lemma~\ref{lem:equal_in_trivial}, we have 

\[
\delta_{\opt}(\abr^{k},\diffp_g) = \sup_{t \in [0,\diffp]}\left\{ q_{t} (1 - e^{\diffp_g - t}) + (1 - q_{t}) \delta_{\opt}(\abr^{2:k},\diffp_g + \diffp - t)
\right\}.
\]

We can then apply our inductive hypothesis because $k - 1\geq 1$ and $\diffp_g + \diffp - t \leq -(k-2)\diffp$, and therefore
\[
\delta_{\opt}(\abr^{k-1},\diffp_g + \diffp - t) =
1 - e^{\diffp_g + \diffp - t} +
\sup_{t_i \in [0,\diffp]} 
\prod_{i=1}^{k-1}(1 - q_{t_i}) 
\left(e^{\diffp_g + \diffp - t + (k-1)\diffp - \sum_{i=1}^{k-1} t_i} - 1\right) 
\]

Plugging in this term and once again using the fact that $q_{t} (1-e^{\diffp_g - t})  + (1 - q_{t})(1 - e^{\diffp_g + \diffp - t}) = 1 - e^{\diffp_g}$ for any $t$, gives our desired equality.
\end{proof}